\documentclass[reqno,11pt]{amsart}

\IfFileExists{mymtpro2.sty}{%
\usepackage[subscriptcorrection]{mymtpro2}}%
{%
}

\usepackage{a4,amssymb}

\newtheorem{theorem}{Theorem}[section]
\newtheorem{lemma}{Lemma}[section]

\newtheorem{proposition}{Proposition}[section]
\theoremstyle{definition}

\newcommand{\bel}{\begin{equation} \label}
\newcommand{\ee}{\end{equation}}

\newcommand{\vecx}{{\bf x}}

\newcommand{\vecz}{{\bf z}}

\newcommand{\vecp}{{\bf p}}
\newcommand{\vecpht}{{\bf \hat{p}}}

\newcommand{\veca}{{\bf a}}

\newcommand{\vecxi}{\mbox{\boldmath$\xi$}}

\newcommand{\vecLambda}{\mbox{\boldmath$\Lambda$}}

\newcommand{\vecell}{\mbox{\boldmath$\ell$}}
\newcommand{\vecL}{{\bf L}}

\theoremstyle{remark}
\newtheorem{remark}{Remark}

\newenvironment{acknowledgement}{\par\medskip\noindent\emph{Acknowledgment.}}

\renewcommand\d{\mathrm d}

\newcommand{\beq}{\begin{equation}}
\newcommand{\eeq}{\end{equation}}
\newcommand{\ba}{\begin{array}}
\newcommand{\ea}{\end{array}}
\newcommand{\bea}{\begin{eqnarray}}
\newcommand{\eea}{\end{eqnarray}}

\newcommand{\R}{\mathbb{R}}

\newcommand{\Sp}{\mathbb{S}}
\newcommand{\T}{\mathbb{T}}

\newcommand{\Z}{\mathbb{Z}}

\newcommand{\N}{\mathbb{N}}

\newcommand{\C}{\mathbb{C}}
\newcommand{\vecai}{{\bf a}}
\newcommand{\vecbi}{{\bf b}}
\newcommand{\vece}{{\bf e}}
\newcommand{\z}{{\bf z}}

\begin{document}

\title[Eigenvalue Clusters for the hydrogen Zeeman Hamiltonian]{Semiclassical Szeg\"o limit of eigenvalue clusters for the hydrogen atom Zeeman Hamiltonian}

\author[M.\ Avenda\~ {n}o-Camacho]{Misael Avenda\~ {n}o-Camacho}
\address{CONACYT Research Fellow -  Departamento de Matem\'aticas, Universidad de Sonora,}
\email{misaelave@mat.uson.mx}

\author[P.\ D.\ Hislop]{Peter D.\ Hislop}
\address{Department of Mathematics,
    University of Kentucky,
    Lexington, Kentucky  40506-0027, USA}
\email{peter.hislop@uky.edu}

\author[C.\ Villegas-Blas]{Carlos Villegas-Blas}
\address{Universidad Nacional Aut\'onoma de M\'exico, Instituto de Matem\'aticas,
Unidad Cuernavaca, Mexico} \email{villegas@matcuer.unam.mx}


\vspace{.1in}

\begin{abstract}
We prove a limiting eigenvalue distribution theorem (LEDT) for suitably scaled
eigenvalue clusters around the discrete negative eigenvalues of the hydrogen atom Hamiltonian
formed by the perturbation by a weak constant magnetic field.
We study the hydrogen atom  Zeeman Hamiltonian
$H_V(h,B) = (1/2)( - i h {\mathbf \nabla} - {\mathbf A}(h))^2 - |x|^{-1}$, defined on $L^2 (\R^3)$,
   in a constant magnetic field ${\mathbf B}(h) = {\mathbf \nabla} \times {\mathbf A}(h)=(0,0,\epsilon(h)B)$
in the weak field limit $\epsilon(h) \rightarrow 0$ as $h\rightarrow{0}$.
We consider the
Planck's parameter $h$
taking  values along the sequence  $h=1/(N+1)$, with $N=0,1,2,\ldots$, and
$N\rightarrow\infty$.
We prove a semiclassical $N \rightarrow \infty$ LEDT  of the Szeg\"o-type
for the scaled eigenvalue shifts
and obtain both ({\bf i}) an expression involving the regularized classical Kepler orbits with energy $E=-1/2$ and ({\bf ii}) a weak limit measure that involves the component  $\ell_3$ of
the angular momentum vector  in the direction of the magnetic field.
This LEDT extends  results of Szeg\"o-type for eigenvalue
clusters for bounded perturbations of the hydrogen atom to the Zeeman effect.
The new aspect of this work is that the perturbation involves
the unbounded, first-order, partial differential operator $w(h, B) = \frac{(\epsilon(h)B)^2}{8} (x_1^2 + x_2^2) - \frac{ \epsilon(h)B}{2} hL_3 ,$ where the operator  $hL_3$ is the third component of the usual angular momentum operator and is
the quantization of $\ell_3$.
The unbounded Zeeman perturbation is controlled
using localization properties of both the hydrogen atom coherent states  $\Psi_{\alpha,N}$, and their derivatives
$L_3(h)\Psi_{\alpha,N}$, in the large quantum number regime $N\rightarrow\infty$.
\end{abstract}

\maketitle
\thispagestyle{empty}


\tableofcontents


\section{Introduction: Limiting eigenvalue distribution theorems.}\label{sec:szego-intro}

The behavior of eigenvalue clusters resulting from the perturbation
of highly degenerate eigenvalues of elliptic operators on compact
manifolds has been studied by many researchers, notably by V.\
Guillemin \cite{guillemin1,guillemin2}, and by A.\ Weinstein \cite{weinstein}.

This is the second paper in which we study the behavior of resonance and eigenvalue clusters
associated with a fixed eigenvalue of a family of
hydrogen atom Hamiltonians, labeled by the Planck's parameter $h$, under perturbations by
external electric and magnetic fields in the weak field limit. In \cite{hVB1}, we studied the resonance cluster associated with the hydrogen atom Stark Hamiltonian in the small electric field regime and proved a Szeg\"o-type result for the resonance shifts. In this paper, we treat the eigenvalue clusters formed by a magnetic field
(the Zeeman effect).

The behavior of eigenvalue clusters for smooth real-valued perturbations $V$ of the Laplacian
on rank one symmetric spaces was studied by V.\ Guillemin
\cite{guillemin1,guillemin2}.
A.\ Weinstein \cite{weinstein} established a \textbf{limiting eigenvalue distribution theorem} (LEDT) for the Laplacian  $\Delta_{M}$ on a compact Riemannian manifold $M$ all of whose geodesics are closed perturbed by a smooth real-valued potential  $V$.  The spectrum of the Laplacian $\Delta_{M}$  consists of eigenvalues $E_N$ with multiplicity $d_N$ that grows polynomially with $N$.
In this case, the semiclassical parameter is the index $N$ of the unperturbed eigenvalue $E_N$.
Since $V$ is bounded, the spectrum of $\Delta_{M}+V$ consists of eigenvalues that form clusters
around the unperturbed eigenvalues.

To explain the LEDT introduced by Weinstein  \cite{weinstein}, we denote by $E_{N,j}$, $j=1,\ldots,d_N$, the eigenvalues in the cluster around $E_N$, and by $\nu_{N,j}=E_{N,j}-E_N$ the eigenvalue shifts.
The LEDT states that for a continuous, real-valued function
$\rho:\R\rightarrow{\R}$ we have
\beq\label{LEDT}
\lim_{N \rightarrow \infty} \frac{1}{d_N} \sum_{j=1}^{d_N} \rho
\left( {\nu}_{N,j} \right) =
\int_{\Gamma} \rho (\hat{V} (\gamma)) ~d \mu_{\Gamma} (\gamma) , \eeq
where $\Gamma$ denotes the space of oriented geodesics of $M$. These are the
classical orbits of the unperturbed problem describing a particle moving on the manifold $M$ with no potential.
The function $\hat{V}:\Gamma\rightarrow{\R}$ is the Radon transform of the potential $V$. Namely, $\hat{V} (\gamma)$ denotes the average of V along the geodesic
$\gamma$ parameterized with respect to arc-length. The measure $\d{\mu}_\Gamma$ is the normalized
measure on the space $\Gamma$ obtained from the restriction to the unit cotangent bundle $T^*_1M$ of the Liouville measure associated to the canonical simplectic form on the symplectic manifold $T^*M$. Weinstein actually proves a LEDT for perturbations given by pseudo-differential operators of order zero.  Here we only state the result for multiplicative potentials  for simplicity.

R.\ Brummelhuis and A.\ Uribe \cite{B-U} extended these results to the study of the semiclassical Schr\"odinger
operator $H_h (V) = - h^2 \Delta + V$ on $L^2 (\R^n)$. The potential $V \geq 0$ is smooth with $V_\infty \equiv
\liminf_{|x| \rightarrow \infty} V(x) > 0$. They studied the semiclassical behavior of the eigenvalue cluster near an energy $0 < E^2 < V_\infty$. They proved an asymptotic expansion of $Tr \rho[ (H_h(V)^{1/2} - E) h^{-1}]$ as $h \rightarrow 0$ and related the coefficients to the classical flow for $|\vecp|^2 +V$ on the energy surface $E$.
For other results on the clustering of eigenvalues for $h$-pseudodifferential operators with periodic flows see, for example, \cite{helffer-sj1,dozias}.

A.\ Uribe and C.\ Villegas-Blas \cite{uribe-villegas1}
extended these results by considering perturbations of the family of hydrogen atom Hamiltonians
$H_V(h) = - (1/2)h^2 \Delta - |x|^{-1}$ defined on $L^2 (\R^n)$, $n\geq{2}$,
by operators of the form $\epsilon (h) Q_h$ where
$Q_h$ is a zero-order pseudo-differential operator uniformly bounded in $h$
and $\epsilon (h) = \mathcal{O}(h^{1+\delta})$, for $\delta >0$.
Here and in the sequel $V$ denotes the Coulomb potential  $V= - |x|^{-1}$.
The spectrum of $H_V(h)$ consists of discrete eigenvalues $E_k(h)=-\frac{1}{2h^2\left(k+\left(\frac{n-1}{2}\right)\right)^2}$ ,
$k\in\N^*$, with  multiplicity   $d_k=O(k^{n-1})$
  together with  the continuous spectrum $[0,\infty)$. Here, we denote by $\N^*$ the set of non-negative integers.
They considered the Planck's parameter taking values along the following sequence converging to zero: $h=1/(N+\frac{n-1}{2})$ with $N\in\N^*$.  Thus for $N$ given and taking
$k=N$, we have that the number
$E=E_{k=N}(h=1/(N+\frac{n-1}{2}))=-1/2$  is an eigenvalue of $H_V(h=1/(N+\frac{n-1}{2}))$ with multiplicity $d_N=O(N^{n-1})$.

In this setting, Uribe and Villegas-Blas \cite{uribe-villegas1} established
a LEDT similar to formula (\ref{LEDT}) but  with ({\bf i})   the  clusters of eigenvalues around the number $E=-1/2$,
{\bf (ii)} the eigenvalue shifts scaled by $\epsilon (h)$,{\bf (iii)} the right-hand side involving
averages of the principal symbol of $Q_h$ along the classical orbits of the regularized Kepler problem on the
energy surface  \beq \Sigma(-1/2)=\left\{({\bf x}, {\bf p})\in\T^*(R^3-\{0\}) \; \arrowvert \; \frac{|{\bf p}|^2}{2}-\frac{1}{|{\bf x}|}=-1/2\right\}, \eeq and {\bf (iv)} the integration  being  with respect to the normalized Liouville measure on the
energy  surface $\Sigma(-1/2)$.

The semiclassical limit is achieved by taking $N\rightarrow\infty$ or equivalently $h\rightarrow{0}$.
A novelty comes in the work of Uribe and Villegas-Blas \cite{uribe-villegas1} from the fact that for a fixed negative energy,
there are two types of (phase space) classical orbits for the classical Hamiltonian flow with Hamiltonian
$G(\vecx, \vecp) = \frac{|{\bf p}|^2}{2} - \frac{1}{|{\bf x}|}$ (the Kepler problem).
Namely,  (i) bounded periodic orbits corresponding to nonzero angular momentum,
and (ii) unbounded collision orbits with zero angular momentum.
Uribe and Villegas-Blas used Moser's regularization of collision orbits (see Appendix 1, section \ref{sec:moser-map})
so that all the collision orbits on $\Sigma(-1/2)$ can be  considered periodic orbits after a time re-parametrization. In this regularization, all orbits on $\Sigma(-1/2)$
correspond to geodesics on the sphere $\Sp^n$ through the Moser map
${\mathcal M}:T^*({\mathbb R}^3)\rightarrow{T}^*({\Sp}^3_o)$ with   $\Sp^3_o$ denoting  the 3-sphere
with the north pole removed (see Appendix 1, section \ref{sec:moser-map} for a review of ${\mathcal M}$). Those passing through the north pole are the collision orbits.
The geodesics on $\Sp^n$ are parameterized by the quotient of  the subset $\mathcal{A}= \{ \alpha\in{\C}^{n+1} ~|~  ~| \Re \alpha | = | \Im \alpha | = 1, ~ \Re \alpha \cdot \Im
\alpha = 0 \}$ of the null quadric in $\C^{n+1}$ with respect to the circle action
(see Appendix 2, section \ref{sec:app2-coherentstates1}).
The set $\mathcal{A}$ corresponds to the unit cotangent bundle ${T}_1^*\Sp^n$ of $\Sp^n$ under the map $\sigma:{\mathcal A}\rightarrow{T}_1^*\Sp^n$  with $\sigma(\alpha)=(\Re \alpha,-\Im \alpha)$.

In this paper, we extend these results to eigenvalue clusters of the hydrogen atom Zeeman Hamiltonian defined on $L^2(\R^3)$. We prove a
LEDT on the semiclassical behavior of the distribution of
the eigenvalue shifts. To explain this in more detail, let us
consider the same setting as Uribe and Villegas-Blas. We regard $E=-1/2$ as an eigenvalue of the family of hydrogen atom Hamiltonians  $H_V(h=1/(N+1))$ with multiplicity $d_N=(N+1)^2$. We consider the atom in an external, constant magnetic field
${\mathbf B}(h) = (0,0,\epsilon (h)B)$, with the constant $B \geq 0$ and $\epsilon (h)=h^{K+\delta}$, $\delta>0$, for some suitably chosen
$K > 0$, see Theorem \ref{thm:main1}.  We consider ${\mathbf B}(h) \rightarrow 0$ as $h \rightarrow 0$.
We refer to this as the weak field limit.

The resulting Hamiltonian
\bea
H_V(h,B) &=&  (1/2)( - i h {\mathbf \nabla} - {\mathbf A}(h))^2 - |x|^{-1} \noindent \nonumber \\
&=& H_V(h) + \frac{(\epsilon (h)B)^2}{8}(x_1^2 + x_2^2) -
\frac{\epsilon (h)B}{2}h L_3
\eea
is called the {\it hydrogen atom Zeeman Hamiltonian}.
Here, ${\bf x} = (x_1,x_2,x_3)$ denotes Cartesian coordinates for ${\mathbb R}^3$,
${\mathbf \nabla}=(\frac{\partial\phantom{x_1}}{\partial{x_1}},\frac{\partial\phantom{x_1}}{\partial{x_2}},\frac{\partial\phantom{x_1}}{\partial{x_3}})$
and
the operator $hL_3=-\imath{h}\left(x_1\frac{\partial\phantom{y}}{\partial{x_2}}-
x_2\frac{\partial\phantom{y}}{\partial{x_1}}\right)$ is the third component of the angular momentum operator $h\vecL=\vecx{\times}({-\imath{h}){\mathbf \nabla}}$. We are working in the symmetric gauge for the vector potential  ${\mathbf A}(h)=\frac{\epsilon (h)B}{2}(-x_2,x_1,0)$.

Although the perturbation  $\frac{(\epsilon (h)B)^2}{8}(x_1^2 + x_2^2) -
\frac{\epsilon (h)B}{2} hL_3$ is not bounded,
we still have an eigenvalue stability theorem that follows from the work of J.\ Avron, I.\ Herbst, and B.\ Simon \cite{AhS1}.
Following \cite{AhS1}, we show that
under the perturbation by the effective magnetic field $\epsilon (h)B$, the
eigenvalue $E=-1/2$ gives rise to a cluster of nearby eigenvalues $E_{N,j}(h,B)=E_{N,j}(1/(N+1), B),
j= 1, \ldots , d_N$, with total geometric multiplicity equal to $d_N$ (see Theorem \ref{thm:stability1} on eigenvalue stability).
We obtain explicit relative bounds on the perturbation that allow an estimate on the size of the cluster.
Our main result is the following LEDT for this eigenvalue cluster in the large $N$ limit corresponding to a weak magnetic field:

\begin{theorem}\label{thm:main1}
Let $B > 0$ be fixed, and let
$\rho$ be a continuous function on $\R$. Let
$\epsilon (h) = h^{33/2 + \delta}$, for some $\delta > 0$, and take $h = 1/(N+1)$, with $N \in \N^*$.
For the eigenvalue cluster $\{ E_{N,j}(1/(N+1),B)\}$, with $j = 1,2, \ldots, d_N$, near $E_N( 1/(N+1) ) = -1/2$,
we have
\bea\label{eq:szego1}
\lefteqn{\lim_{N \rightarrow \infty} \frac{1}{d_N}
\sum_{j=1}^{d_N} \rho \left( \frac{ E_{N,j} (1/(N+1), B) - E_N (1/(N+1)) }{ \epsilon
(1/(N+1)) } \right)}  \nonumber \\
  && \phantom{xxxxxxxxxxxxxxxxx} = \int_{\Sigma (-1/2) } \rho \left( -\frac{B}{2} \ell_3 ({\bf x, p})
  \right) ~ d\mu_L ({\bf x, p}), \label{eq-main-thm}
 \eea
where $\ell_3 ({\bf x,p}) = x _1p_2 - x_2 p_1$ is the third component of the classical angular momentum vector $\vecell=\vecx{\times}\vecp$
on the energy surface $\Sigma ( -1/2)$
with collision orbits treated as in \cite{uribe-villegas1}.
The measure $d\mu_L$ is the normalized restriction of the
Liouville measure   to the energy surface $\Sigma (-1/2)$.  Here $\T^*(R^3-\{{\bf 0}\})$ is endowed with its  canonical symplectic form.
\end{theorem}

We recall the proof in Appendix 1, section \ref{sec:moser-map}, that  $\ell_3$ is conserved and bounded on the Kepler orbits on the energy surface $\Sigma (-1/2)$, so there is no problem working with arbitrary continuous functions.

Theorem \ref{thm:main1} parallels and extends the result of Uribe and Villegas-Blas \cite{uribe-villegas1} on
eigenvalue clusters formed by bounded perturbations $Q_h$ of the hydrogen atom Hamiltonian.
These ideas were applied by two of the authors \cite{hVB1} to the hydrogen atom Stark Hamiltonian. The main result of \cite{hVB1} is a limiting resonance distribution  theorem for resonance clusters associated with a hydrogen atom eigenvalue under the Stark perturbation by an external electric field. Theorem \ref{thm:main1} considers the case of  eigenvalue clusters formed when a hydrogen atom is placed in a constant magnetic field.
Both of these works have in common an unbounded perturbation.
As in \cite{hVB1}, control of the unbounded Zeeman perturbation is obtained through localization properties of coherent states of the hydrogen atom Hamiltonian. However, since the Zeeman perturbation is a first order differential operator, we have to extend these localization results to the derivatives of the coherent states.
This requires an additional analysis (see section \ref{subsec:ang-mom-term}).

We remark that the size of the exponent $K$ in Theorem \ref{thm:main1} is far from optimal. Roughly speaking,
if we suppose that the perturbation proportional to $(x_1^2 + x_2^2)$ is bounded, the size of the eigenvalue cluster around the eigenvalue $-1/2$ is $N^{-K}$. 
For the eigenvalue clusters to be well separated, we need $N^{-K} \approx N^{-1}$, so $K > 1$.
But, the perturbation is unbounded, and this forces us to take $K$ much larger
in order to control the error in the estimate of the difference of resolvents in Lemma \ref{lemma:convergence-norm1} and Theorem \ref{thm:stability1}.

Let $d\mu$ be defined as the normalized SO(4)-invariant measure on ${\mathcal A} \subset \C^{n+1}$ defined above.
In Proposition \ref{liouville},  section \ref{sec:poly1},   we show that
the Liouville measure
$d\mu_L$ on the energy surface $\Sigma(-1/2)$ is the push-forward measure of $d\tilde\mu$
by the map ${\mathcal M}^{-1}\circ\sigma$, where $d\tilde\mu$ is the
restriction of $d\mu$ to a subset $\tilde{{\mathcal A}}$ of ${\mathcal A}$ with
$\mu({\mathcal A}-\tilde{{\mathcal A}})=0$, see equation (\ref{liou-quadric}) .
Thus the  right hand side of (\ref{eq-main-thm}) can be written in terms of an integral over ${\mathcal A}$.
This allows the following reformulation of Theorem \ref{thm:main1}.

\begin{theorem}\label{thm:main2}  Under the same hypothesis as in Theorem \ref{thm:main1},
we have
\bea\label{eq:szego2}
\lefteqn{\lim_{N \rightarrow \infty} \frac{1}{d_N}
\sum_{j=1}^{d_N} \rho \left( \frac{ E_{N,j} (1/(N+1), B) - E_N (1/(N+1)) }{ \epsilon
(1/(N+1)) } \right)}  \nonumber \\
  && \phantom{xxxxxxxxxxxxxxxxx} = \int_{{\mathcal A}} \rho \left( -\frac{B}{2} \ell_3 (\alpha)
  \right) ~ d\mu(\alpha), \label{eq-main-thm2}
 \eea
where, for all $\alpha \in \mathcal{A}$, we have $\ell_3(\alpha) = (\Re \alpha)_1 (\Im \alpha)_2- (\Re \alpha)_2 (\Im \alpha)_1$, the third component of the classical angular momentum associated to $\alpha$.
The function  $\ell_3(\alpha)$ can be thought of as a continuous extension of the assignment
$\alpha\rightarrow({\bf x,p})\mapsto{\ell_3({\bf x,p})}$ thorough the map  ${\mathcal M}^{-1}\circ\sigma(\alpha)$
which is well defined as long as $\Re \alpha$ is not the north pole of $\Sp^3$.
 \end{theorem}

We can think of the right-hand side of \eqref{eq-main-thm2} as  a linear positive functional on
$C_0^\infty(\R)$.  By a Riesz Representation Theorem, there exists a measure $d\kappa$ on the real line such that the right-hand side of (\ref{eq-main-thm2}) can be written as the integral of $\rho$ with respect to $d\kappa$.  The measure $d\kappa$ can be seen as the push-forward  measure of $d\mu(\alpha)$ under the map $-\frac{B}{2}\ell_3:{\mathcal A}\rightarrow\R$. By using an explicit expression for
$d\mu$ in terms of coordinates for both the classical angular momentum vector $\vecell$  and the
Runge-Lenz vector ${\mathbf a}$, and the relative angle between the position vector ${\bf x}$ and ${\mathbf a}/{\mathbf a}|$ (see \cite{villegas}), we can actually provide an explicit expression for $d\kappa$. This leads to another formulation of Theorem \ref{thm:main1} :

\begin{theorem}\label{thm:main3}  Under the same hypothesis as in Theorem \ref{thm:main1},  we have:
\bea\label{eq:szego3}
\lefteqn{\lim_{N \rightarrow \infty} \frac{1}{d_N}
\sum_{j=1}^{d_N} \rho \left( \frac{ E_{N,j} (1/(N+1), B) - E_N (1/(N+1)) }{ \epsilon
(1/(N+1)) } \right)}  \nonumber \\
  && \phantom{xxxxxxxxxxxxxxxxx} = \int_{[-1,1]} \rho \left(-\frac{B}{2} u
  \right) (1-|u|)~ du, \label{eq-main-thm3}
 \eea
 where $du$ denotes the Lebesgue measure on the interval $[-1,1]$.  The variable $u$
 can be thought of as the third component $\ell_3$ of the classical angular momentum vector $\vecell$.
\end{theorem}

The measure $d\kappa= \frac{2}{B}(1-\frac{2}{B}|x|)~ dx$ is supported on the interval $[-\frac{B}{2},\frac{B}{2}]$ (for $B>0$), where $dx$ is the Lebesgue measure on $\R$. This measure gives us a precise picture of how, for $N$  large, the scaled eigenvalue shifts
$$
\frac{ E_{N,j} (1/(N+1), B) - E_N (1/(N+1)) }{ \epsilon
(1/(N+1)) } , \;\; j=1,\ldots,d_N ,
$$
are distributed in the interval $[-\frac{B}{2},\frac{B}{2}]$.
The distribution around the origin in the interval
 $[-\frac{B}{2},\frac{B}{2}]$ is determined according to the probability density function
 $P(x)= \frac{2}{B}(1-\frac{2}{B}|x|)$.

Theorems \ref{thm:main1}, \ref{thm:main2} and \ref{thm:main3} give a rather complete analytical and geometric description of the limiting eigenvalue distribution for the eigenvalue clusters formed by the Zeeman perturbation of the hydrogen atom Hamiltonian.

We remark that Theorem \ref{thm:main3}  can actually
be shown in a different way than using Theorem \ref{thm:main2}
and the expression for $d\mu$ mentioned above. One can use a suitable eigenvalue approximation for the cluster   around $E_N (1/(N+1))$ and then evaluate the left hand side of \eqref{eq-main-thm3} by means of Riemann sums. This is shown in section \ref{sec:eigen-approx}. This procedure, however, completely masks the beautiful geometric foundations of the problem appearing in Theorems \ref{thm:main1}  and \ref{thm:main2}.

\subsection{Contents.}\label{subsec:contents1}

In section \ref{sec:scale1}, we scale the hydrogen atom Zeeman Hamiltonian using the dilation group.
This establishes a countable family of
scaled hydrogen atom Zeeman Hamiltonians  $S_V(\lambda) = S_V + W(\lambda)$. The operator $S_V$
is a fixed, $h$-independent, hydrogen atom Hamiltonian $S_V = -\frac{1}{2}\Delta - \frac{1}{|x|}$.
The magnetic perturbation is
$W(\lambda) = \frac{\lambda ^2}{8} (x_1^2 + x_2^2) - \frac{\lambda }{2} L_3 $, where the effective magnetic field strength is $\lambda(h,B)=h^3 \epsilon (h) B$,
and $L_3=-\imath\left(x_1\frac{\partial\phantom{y}}{\partial{x_2}}-
x_2\frac{\partial\phantom{y}}{\partial{x_1}}\right)$.
In this new framework, we want to establish conditions on the size of $\epsilon(h)$ in order to show the existence of clusters of eigenvalues around $E_N=-\frac{1}{2(N+ 1)^2}$, for N sufficiently large, with $h=1/(N+1)$.

In section \ref{sec:spectrum1}, we provide the description of known results on the spectrum of the operator $S_V(\lambda)$, with $\lambda$ fixed, based on references \cite{AhS1} and \cite{FW}. Then we mention and prove a key result of Avron, Herbst, and Simon
\cite{AhS1,AhS3} on the norm convergence $V (S_0 (\lambda) - \z)^{-1} \rightarrow V(S_0 - \z)^{-1}$ as $\lambda \rightarrow {0}$ for $\z \not\in [0, \infty)$,  with
\beq\label{defS0}
S_0=-\frac{1}{2}\Delta  \,\,\,\,\,\,\, {\rm and}  \,\,\,\,\,\,\, S_0 (\lambda)=S_0+W(\lambda),
\eeq
by presenting  several important resolvent estimates necessary for our work. Moreover, we study such a rate of convergence with
 $\lambda(h,B)=h^3 \epsilon (h) B$ and  $h=1/(N+1)$ when $N \rightarrow \infty$ and $\z$ is in a  circle of radius $O(N^{-3})$ with center $E_N$.
Then we are able to show an eigenvalue stability theorem, Theorem \ref{thm:stability1}, by estimating the difference between corresponding spectral projectors associated to the perturbed and unperturbed Hamiltonians on a  small disk around $E_N$.

In section \ref{sec:trace1},
we use the stability theorem in order to show that the averages appearing on the left hand side of equation
(\ref{eq-main-thm}) (with a factor $h^2$ included in the denominator due to scaling) can be approximated by the
normalized trace of $\frac{1}{d_N}\rho\left(\Pi_N\left(-\frac{B}{2}hL_3\right)\Pi_N\right)$ with $\Pi_N$ the projector onto the
eigenspace of the unperturbed operator $S_V$ with eigenvalue $E_N$. Next, in section \ref{sec:poly1}, we take the semiclassical limit $N\rightarrow{\infty}$ of this last trace by using the Stone-Weierstrass Theorem, the coherent states for the hydrogen atom introduced in \cite{thomas-villegas1}, and the stationary phase method in order to estimate the
expected value  of $\left(-\frac{B}{2}hL_3\right)^m$, $m\in\N^*$,  between coherent states.  We use decay properties of coherent states shown in \cite{thomas-villegas1} but, in addition,  we need to estimate decay of their derivatives.

Finally, an alternate proof of Theorem \ref{thm:main3} is presented in section \ref{sec:eigen-approx}.

We include two appendices. The Kepler problem and the Moser map are briefly  described in the first appendix in section \ref{sec:moser-map}. In the second appendix, section \ref{sec:app2-coherentstates1}, details of the coherent states for the hydrogen atom are presented.


\begin{acknowledgement}

PDH was partially supported by NSF grants 0803379 and 1103104 during the time this work was done.
CV-B was partially supported by the projects PAPIIT-UNAM IN106812, PAPIIT-UNAM IN104015 and thanks
the members of the Department of Mathematics of the University of Kentucky for their hospitality during a visit.  MA-C was supported by a fellowship of DGAPA-UNAM, by project  PAPIIT-UNAM IN106812, and by CONACYT under the Grants 219631, CB-2013-01 and 258302, CB-2015-01.
\end{acknowledgement}


\section{The basic operators and scaling.}\label{sec:scale1}

The hydrogen atom Hamiltonian $H_V (h)$ with the semiclassical parameter $h$ acts on the dense domain $H^2 ( \R^3)$ in the Hilbert space
$L^2 (\R^3)$. The operator is self-adjoint on this domain and given by
\beq\label{eq:hydro1}
H_V(h) = - \frac{h^2}{2} \Delta - \frac{1}{|x|}.
\eeq
We denote the Coulomb potential by $V(x) = - 1 / |x|$. The discrete spectrum consists of an infinite family of eigenvalues $E_k (h)$
\beq\label{eq:ev1}
E_k(h) = \frac{-1}{2 h^2 (k+1)^2}, ~~k = 0,1,2,\ldots
\eeq
each eigenvalue having multiplicity $d_k := (k+1)^2$. The essential spectrum is $[0, \infty)$.
With the choice of  $h = 1/(N+1)$ and $k = N$, we see that $E_{k=N}(h=1/(N+1)) = -1/2$ is in the spectra of the
countable family of Hamiltonians $H_V(1/(N+1))$, $N \in \N$. The multiplicity of the eigenvalue
$-1/2$ is $d_N=(N+1)^2$.

We next consider a hydrogen atom in a constant magnetic field. We assume, without loss of generality, that the magnetic field
has the form ${\mathbf B}(h) = (0, 0, \epsilon(h)B)$. We keep
$B \geq 0$ fixed and use it only to control whether  the magnetic field is on or not.   We  control the strength of the magnetic field by taking
$\epsilon (h)=h^{K}$, with a constant $K>0$
chosen below.
We choose the gauge such that the vector potential ${\mathbf A}(h)$ is given by
${\mathbf A}(h)=\frac{\epsilon (h)B}{2}(-x_2,x_1,0)$.
The \textbf{unscaled Zeeman hydrogen Hamiltonian} is
\bea\label{eq:zeeman1}
H_V(h,B) &=& \frac{1}{2}(-i h \nabla - {\mathbf A}(h))^2 - \frac{1}{|x|}  \nonumber \\
 &=& H_V(h) + w(h, B).
\eea
where  the \textbf{unscaled Zeeman perturbation} $w(h ,B)$ is given by
\beq\label{eq:zeeman-pert1}
w(h, B) = \frac{(\epsilon(h)B)^2}{8} (x_1^2 + x_2^2) - \frac{\epsilon(h)B}{2} hL_3 ,
\eeq
with $L_3=-\imath\left(x_1\frac{\partial\phantom{y}}{\partial{x_2}}-
x_2\frac{\partial\phantom{y}}{\partial{x}_1}\right)$.

To implement scaling of these Hamiltonians, we use the dilation group $D_\alpha$, $\alpha > 0$. The dilation group
is a representation of the multiplicative group $\R^+$ and
has a unitary implementation on $L^2 ( \R^3)$ given by
\beq\label{eq:dilation0}
(D_\alpha f)(x) = \alpha^{3/2} f(\alpha x),   \;\;\;    f\in{L^2(\R^3)}.
\eeq
Using the relation $D_{\alpha} L_3 D_{\alpha}=L_3$, we scale the Hamiltonian in \eqref{eq:zeeman1} by $\alpha = h^2$:
\bea\label{eq:scale2}
&&D_{h^2} H_V(h, B) D_{h^{-2}} \nonumber \\
 &&\phantom{xxxxxx}=\frac{1}{h^2} \left[ -  \frac{1}{2}\Delta - \frac{1}{|x|} + \frac{\left( h^3 \epsilon (h)B \right)^2}{8} (x_1^2 + x_2^2) - \frac{h^3 \epsilon (h) B }{2} L_3\right] \nonumber \\
 &&  \phantom{xxxxxx}
 =: \frac{1}{h^2} S_V(\lambda(h,B)).
\eea
The \textbf{scaled Zeeman hydrogen  Hamiltonian} $S_V(\lambda(h, B))$  is defined via the
      \textbf{effective magnetic field} $\lambda(h,B)=h^3 \epsilon (h) B$ and the operator
      $S_V(\lambda)$  is given by:
\bea\label{eq:rescale1}
 S_V(\lambda)&=& -  \frac{1}{2}\Delta - \frac{1}{|x|} + \frac{\lambda^2}{8} (x_1^2 + x_2^2) - \frac{\lambda}{2} L_3  \nonumber \\
 & = & S_V  + W(\lambda).
\eea
where we write $S_V\equiv- \frac{1}{2} \Delta - \frac{1}{|x|}$ for the scaled hydrogen atom Hamiltonian and
the magnetic perturbation is
\beq\label{eq:perturbationW}
W(\lambda) = \frac{\lambda ^2}{8} (x_1^2 + x_2^2) - \frac{\lambda }{2} L_3 .
\eeq
The \textbf{scaled Zeeman perturbation} is then given by $W(\lambda(h,B))$.

Note that we can
make the effective magnetic field small $\lambda(h,B)$ by taking $h \rightarrow 0$. Equivalently, we may set
$h = 1 / (N+1)$ and take $N \rightarrow \infty$. For $B=0$, the eigenvalues of $S_V(0) $ are given by $E_k\equiv E_k(1) = - 1/ (2 (k+1)^2)$,
with $k \in \N^{*}$ and multiplicity $d_k=(k+1)^2$.

Since the discrete spectra of the operators $H_V(h, B)$ and $S_V(\lambda(h,B)) $ are the same up to the factor $h^2$,
Theorem \ref{thm:main1} will be proved by establishing a LEDT theorem
 for the family of operators  $S_V(\lambda(h=1/(N+1),B) )$, $ N\in\N^*$ , by studying the eigenvalue distribution in the  cluster around $E_N(1) = - 1/ (2 (N+1)^2)$ and then taking the corresponding limit when $N\rightarrow\infty$. Since the perturbation $W(\lambda(h,B))$ is unbounded,
 the existence of these clusters of eigenvalues is by no means immediate.   A suitable version of a stability theorem due to Avron, Herbst and  Simon \cite{AhS1}, together with an adequate choice of the exponent $K>0$ in the definition of $\epsilon(h)$, guarantee that, for $N$ sufficiently large, the eigenvalue cluster around $E_N$  is  well defined and  the total multiplicity of the eigenvalues in the cluster is $d_N$.
This is the content of Theorem \ref{thm:stability1} proved in the next section.


\section{Spectral analysis of the Zeeman hydrogen atom Hamiltonian and eigenvalue clusters.}\label{sec:spectrum1}

The main goal of this section is to show the existence of eigenvalue clusters ${\mathcal C}_N$ for the operator
$S_V (\lambda(h, B))$ around the unperturbed eigenvalues $E_N=-1/2(N+1)^2$, taking $h=1/(N+1)$ with N sufficiently large.  We will show that there exist circles $\Gamma_N$ with centers $E_N$ and radii $r_N \approx N^{-3}$ such that the total number of eigenvalues of $S_V (\lambda(1/(N+1), B))$
inside ${\Gamma}_N$, including multiplicity, is equal to the multiplicity $d_N=(N+1)^2$ of the eigenvalue $E_N=-1/2(N+1)^2$ of $S_V$.
This fact is a consequence of the main technical result of this section showing  that the norm of the difference of the spectral projectors $P_N$ and $\Pi_N$
associated to the spectrum of  the operators $S_V (\lambda(h, B))$ and $S_V$, respectively, inside ${\Gamma}_N$ is $O(N^{-\sigma})$, $\sigma>0$, and therefore smaller than one for N sufficiently large.  This will give us the eigenvalue stability
that we need in order to have well-defined clusters of eigenvalues.

In subsection \ref{subsec:spectrum1} we describe spectral properties of  $S_V (\lambda(h, B))$  by summarizing
some of the results of Avron, Herbst, and Simon in their papers \cite{AhS1,AhS3}. As we are only concerned with the Coulomb potential, we state their results for this case.

The eigenvalue stability $\|P_N-\Pi_N\| \rightarrow 0$ as $N \rightarrow \infty$ would be immediate if we had norm resolvent convergence
 of $S_V (\lambda(1/(1+N), B))$ to $S_V$ when $N\rightarrow\infty$. However, this is not the case as it was shown in \cite{AhS1}.  Avron, Herbst, and Simon \cite{AhS1} showed that we still can have eigenvalue stability due to the fact that for
$\z \not\in [0, \infty)$ we have the norm convergence $V (S_0 (\lambda) - \z)^{-1} \rightarrow V(S_0 - \z)^{-1} $ as  $\lambda \rightarrow 0$ (see Lemma \ref{lemma:convergence-norm1})  with $S_0$ and $S_0(\lambda)$ given in Eq. (\ref{defS0}). In subsection \ref{subsec:stability1},
we describe the work of Avron, Herbst and Simon about this point by refining  some of their estimates in order to make the dependance on $\lambda(h,B)$ explicit. We prove eigenvalue stability in subsection \ref{subsec:ev-clusters1} by following reference
\cite{AhS1} and prove both suitable and finer estimates required for our purposes.

\subsection{The spectrum.}\label{subsec:spectrum1}

The Hamiltonian obtained from the scaled Zeeman hydrogen Hamiltonian \eqref{eq:rescale1} by setting the Coulomb
potential equal to zero is denoted by $S_0(\lambda(h,B)) $  with $S_0(\lambda)$ given by Eq.(\ref{defS0}) and $\lambda(h,B)=h^3 \epsilon (h) B$.
For $\lambda > 0$, the spectrum of $S_0(\lambda)$ is
purely absolutely continuous and equal to the half line $[\lambda/2, \infty)$.
We note that the operator $S_0(\lambda)$ may be represented as
a tensor product on the space $L^2 (\R^3) = L^2 (\R^2_{x_1,x_2}) \otimes L^2 (\R_{x_3})$. For this purpose,
we recall the two-dimensional Landau Hamiltonian $S_L(\lambda) = -\frac{1}{2} \Delta_{x_1,x_2} + \frac{\lambda^2}{8} (x_1^2 + x_2^2) - \frac{\lambda}{2} L_3$. This operator
has pure point spectrum $E_n(\lambda) = \frac{\lambda}{2}(n+1)$ with  $n\in\N^{*}$. Each Landau level $E_n(\lambda)$ is an eigenvalue of infinite multiplicity. Then, the Hamiltonian $S_0(\lambda)$ may be written as $S_0(\lambda) = S_L(\lambda) \otimes I_1 + I_2 \otimes
\frac{1}{2}(- d^2 / d{x_3}^2)$, where $I_j$ is the identity operator on $L^2 (\R^j)$, $j=1,2$, respectively.
The Landau levels appear as thresholds of the operator $S_0 (\lambda)$.
The spectrum of $S_0(\lambda)$ can then be computed using a well-known result on the spectra of tensor products
\cite[section XIII.9, Theorem XIII.35]{RS4}. It follows directly that $\sigma (S_0(\lambda)
= \{ E \in E_n(\lambda) + \overline{\R^+} ~|~ n \in \N^* \} = [\lambda/2, \infty)$, since $\inf \sigma (S_L(\lambda)) = \lambda / 2$
and the spectrum of $\frac{1}{2}
\left(- d^2 / d{x_3}^2\right)$ is the closed half-line $[0, \infty)$.

We now consider $S_V(\lambda(h,B))$ defined via the operator $S_V(\lambda)$ given in Eq. \eqref{eq:rescale1}. The operator $S_V(\lambda)$ is best understood by studying its restriction
to the eigenspaces of $L_3$. These subspaces are $S_V(\lambda)$-invariant since $S_V(\lambda)$ commutes with $L_3$.
The eigenfunctions of the azimuthal angular momentum operator $L_3$, as an operator on the circle, are $\varphi_m( \phi) = e^{ i m \phi}, m \in \Z$.
We write $\mathcal{H}_m$, $m \in \Z$, for the subspace of $L^2(\R^3)$ consisting of functions whose angular momentum decomposition contain only $\varphi_m (\phi)$.
We then have the direct sum decomposition $L^2 (\R^3) = \bigoplus_{m \in \Z} \mathcal{H}_m$.
The restriction $S^{(m)}(\lambda) \equiv S_V(\lambda) | \mathcal{H}_m$ of $S_V(\lambda)$ to infinite-dimensional subspaces $\mathcal{H}_m, ~m \in \Z$, of constant azimuthal angular momentum $m \in \Z$, has the form
\beq\label{eq:zeeman-1m}
S^{(m)}(\lambda) =\left. \left( - \frac{1}{2}\Delta - 1 / |x| + \frac{\lambda^2}{8} (x_1^2 + x_2^2) - \frac{\lambda}{2} m \right) \right|_{ \mathcal{H}_m} .
\eeq

We let $E^{(m)}(\lambda) \equiv \inf \sigma (S^{(m)}(\lambda))$. This number is a simple isolated eigenvalue of $S^{(m)}(\lambda)$. We refer to $E^{(m)}(\lambda)$ as the ground state of $S^{(m)}(\lambda)$.
For negative indices $m < 0$, these eigenvalues satisfy the relation
\beq\label{eq:zeeman-2m}
E^{(-m)}(\lambda) = E^{(m)} (\lambda) +  m \lambda, ~~\mbox{for} ~m \geq 0 .
\eeq
Each operator $S^{(m)}(\lambda)$ has discrete spectra consisting of \emph{simple} eigenvalues accumulating at the bottom of the essential spectrum (see \cite{FW}, page 5, Main Results, part b).
The essential spectrum of $S^{(m)}(\lambda)$ consists of half-lines
\bea\label{eq:ess=sp-m1}
m < 0 &  \sigma_{\rm ess} (S^{(m)}(\lambda)) = [ (2 |m|+1)\frac{\lambda}{2}, \infty ),  \nonumber  \\
m \geq 0 &  \sigma_{\rm ess} (S^{(m)}(\lambda)) = [ \frac{\lambda}{2}, \infty ).
\eea


The spectrum of $S_V(\lambda)$ is the union of the spectra of $S^{(m)}(\lambda)$, for $m \in \Z$. It follows from \eqref{eq:ess=sp-m1}
that $\frac{\lambda}{2} = \inf \sigma_{\rm ess} (S_V(\lambda))$.
The ground state eigenvalues of the operators $S^{(m)}(\lambda)$, for $m \geq 0$, are strictly ordered:
\beq\label{eq:ev-order1}
E^{(0)}(\lambda) < E^{(1)} (\lambda) < E^{(2)}(\lambda) < \ldots \leq \frac{\lambda}{2} = \inf \sigma_{\rm ess} (S_V(\lambda)).
\eeq
Because of this ordering \eqref{eq:ev-order1} for $m \geq 0$ and the relation \eqref{eq:zeeman-2m},
the ground state of $S_V(\lambda)$ is $E^{(0)}(\lambda) = \inf \sigma (S_V(\lambda))$. It is an isolated eigenvalue
satisfying $E^{(0)}(\lambda) = -1/2 + \mathcal{O}(\lambda)$.
From \eqref{eq:ev-order1}, the discrete spectrum of $S_V(\lambda)$ consists of infinitely-many discrete eigenvalues $\{ E^{(m)}(\lambda) ~|~ m \geq 0 \} \cup
\{ E^{(-m)} (\lambda) ~|~ E^{(m)}(\lambda) < (1 - 2m) \frac{\lambda}{2}, m > 0 \}$, less than $\frac{\lambda}{2}$, accumulating at $\frac{\lambda}{2} = \inf \sigma_{\rm ess} (S_V(\lambda))$.
There are infinitely-many embedded eigenvalues of finite multiplicity in the essential spectrum $[\frac{\lambda}{2}, \infty)$ since, by
\eqref{eq:zeeman-2m}, for $m>0$ large enough, $E^{(-m)} (\lambda) >> \frac{\lambda}{2}$.


\subsection{Norm resolvent estimates and the key lemma.}\label{subsec:stability1}

We present a refined version of Lemma 6.6 of Avron, Herbst, and Simon \cite{AhS1} on the norm convergence $V (S_0 (\lambda) - \z)^{-1} \rightarrow V(S_0 - \z)^{-1}$ as $\lambda \rightarrow 0$ for $\z \not\in [0, \infty)$ that gives the rate of the convergence.
We will specialize to the case of the Coulomb potential $V(x) = - 1 / |x|$ and obtain finer estimates when $\z$ is close to an eigenvalue $E_N$ of the hydrogen atom Hamiltonian $S_V$.
We denote the resolvent of $S_0$ by $R_0 (\z) = (S_0 - \z)^{-1}$, of $S_V(\lambda)$ by $R_{V,\lambda}(\z) \equiv ( S_V(\lambda) - \z)^{-1}$, so that for $V = 0$, we have $R_{0,\lambda}(\z) = ( S_0(\lambda) - \z)^{-1}$.
The spectra of $S_0$ and $S_0(\lambda)$ lie in the positive half-line, so both resolvents $R_0 (\z)$ and $R_{0,\lambda}(\z)$
exist as bounded operators for  $\z \not\in [0, \infty)$. We have the basic bounds of their norms:
\bea \label{eq:basic-norm-bounds1}
\| R_0(\z) \| &\leq & [{\rm dist} (\z, [0, \infty))]^{-1} , \nonumber  \\
\| R_{0, \lambda}(\z) \| & \leq & [{\rm dist} (\z, [\lambda/2, \infty))]^{-1} .
\eea

Avron, Herbst, and Simon \cite[Lemma 6.4]{AhS1} proved that for $\z \not\in [0, \infty)$,
$R_{0,\lambda}(\z)$ converges strongly to $R_0 (\z)$ ar $\lambda \rightarrow 0$.
Moreover, in \cite[Theorem 6.3]{AhS1}, they showed that
    $S_V(\lambda)=S_0(\lambda)+V$ does not converge in the norm resolvent sense to  $S_V=S_0+V$ as $\lambda \rightarrow 0$, which includes the fact that $R_{0,\lambda}(\z)$ does not converge to $R_0 (\z)$ in norm as $\lambda \rightarrow 0$.
However, they show \cite[Lemma 6.6]{AhS1} the norm convergence $V (S_0 (\lambda) - \z)^{-1} \rightarrow V(S_0 - \z)^{-1}$ as $\lambda \rightarrow 0$ for $\z \not\in [0, \infty)$, which plays the key role in the proof of eigenvalue stability.
We prove this last result in Lemma \ref{lemma:convergence-norm1} and obtain an estimate on the rate of convergence
necessary in the proof of the eigenvalue stability theorem, Theorem \ref{thm:stability1}.

In order to prove this last result, we introduce the   cut-off function $\chi_R $ as the characteristic function of the unit  $B_R(0)$ of radius $R>0$ centered at the origin.
In the sequel, the symbol $C$ will denote a constant whose value may differ from line-to-line but is independent of $N$. The first part of the following lemma is effectively \cite[Lemma 6.6]{AhS1} and the second part gives the rate of convergence.  The following notation will be used in the sequel:  a bounded  operator whose norm is  $O(N^\alpha)$, for some $\alpha\in\R$, will be denoted by  $O(N^\alpha)$ as well.

\begin{lemma}(Key Lemma)\label{lemma:convergence-norm1}
Consider $\z \not\in [0, \infty)$.
\begin{enumerate}

\item We have the following convergence in norm:
\beq\label{eq:norm1}
V (S_0 (\lambda) - \z)^{-1} \rightarrow V(S_0 - \z)^{-1} ,
\eeq
as $\lambda \rightarrow 0$.

\item Consider $\lambda=\lambda(h)$ with  $h=1/(N+1)$ and $\epsilon(h)=h^q$, $q>3/2$. For $|\z-E_N|=O(N^{-3})$ we have
\beq\label{eq:norm2}
V (S_0 (\lambda(h)) - \z)^{-1} - V(S_0 - \z)^{-1} = O\left(N^{-\left(\frac{2q - 3}{5}\right)}\right)  ,
\eeq
as $N\rightarrow\infty$.
\end{enumerate}
\end{lemma}

\begin{proof}
1.  For any fixed $R>0$, we decompose $V=-1/|x|$ as $V=V_1+V_2$, with $V_1=V\chi_R$ and $V_2=V(1-\chi_R)$, so that $V_1$ has compact support and $V_2$ is bounded. We choose $R>0$ below.
The contribution of $V_2$ to \eqref{eq:norm1} is easy to treat.
For $\z \not\in [0, \infty)$, both
$ (S_0 (\lambda) - \z)^{-1}$ and  $(S_0 - \z)^{-1}$ are bounded by $1/d(\z)$  with
$d(\z)\equiv{{\rm dist} (\z, [0, \infty))}$ . Thus
the contribution to \eqref{eq:norm1} from $V_2$ is bounded by
\beq
\|V_2\left[(S_0 (\lambda) - \z)^{-1} - (S_0 - \z)^{-1}\right]\| \leq \|V_2\|_\infty\frac{2}{d(\z)} =
\frac{2}{Rd(\z)}. \label{estV_2}
\eeq
We note that the contribution \eqref{estV_2} vanishes as $R \rightarrow \infty$ for $d(\z)$ fixed.

\noindent
2. As for the contribution of $V_1$ to \eqref{eq:norm1}, we write the difference of the resolvents
using (\ref{defS0}) as
\bea
V_1 \left[(S_0 (\lambda) - \z)^{-1} - (S_0 - \z)^{-1}\right] &=&
 \lambda{V_1}(S_0-\z)^{-1}{\vecLambda}\cdot\left({\vecpht}-\lambda{\vecLambda}\right)  (S_0(\lambda)-\z)^{-1}  \nonumber \\
&&+ \frac{\lambda^2}{2}V_1(S_0-\z)^{-1}
{|\vecLambda|}^2(S_0(\lambda)-\z)^{-1} \label{eqV_1}
\eea
where ${\vecLambda}=\frac{1}{2}(-x_2,x_1,0)$ and $\vecpht=(\hat{p}_1,\hat{p}_2,\hat{p}_3)=-\imath{\mathbf \nabla}$.
Since the spectrum of
$S_0(\lambda)=\frac{1}{2}\left(\vecpht-\lambda\vecLambda\right) ^2$ lies in
the interval $[0,\infty)$, the operators $\left(\vecpht-\lambda{\vecLambda}\right)  (S_0(\lambda)-\z)^{-1}$
 and $(S_0(\lambda)-\z)^{-1}$ are uniformly bounded for $\lambda\geq{0}$ and $\z \not\in [0, \infty)$:
\bea
\| \left(\vecpht-\lambda{\vecLambda}\right)  (S_0(\lambda)-\z)^{-1} \|&\leq&
{\rm max}\{s/|s^2/2-\z| \; | \; s\in[0,\infty)\} \nonumber \\
&\leq& 1/\sqrt{G(\z)},  \;\;\;\;\;\;\;    \label{eq-bd1}\\
\|(S_0(\lambda)-\z)^{-1} \|&\leq& 1/d(\z)  \label{eq-bd2}
\eea
with $G(\z)\equiv|\z|-\Re(\z)$.

\noindent
3. We show that ${V_1}(S_0-\z)^{-1}{\vecLambda}$ and $V_1(S_0-\z)^{-1}
|{\vecLambda}|^2$ are bounded operators in order to prove that the norm of the
left hand side of
(\ref{eqV_1}) goes to zero as $\lambda \rightarrow 0$ with $R$ fixed.  We also estimate the rate of convergence.
We use the following estimate \cite{DO} that is a consequence of
the bound $\| \Psi \|_\infty \leq C \| \Psi \|_{H^2}$, for all $\Psi \in H^2 (\R^3)$.
For $\beta>0$,  there exist a constant C such that
   for $\Psi \in {H}^2(\R^3)$ we have
 \beq\label{estimate-inf}
 \|\Psi\|_{\infty} \leq
 C \left[\frac{1}{\beta^{1/2}} \|\Delta\Psi \| _2
 + \beta^{3/2} \|\Psi\|_2 \right].
 \eeq
We use this inequality together with the simple bound
\beq\label{estimate-inf2}
\| f \Psi \|_2 \leq \|f \|_2 \| \Psi \|_\infty ,
\eeq
for $f \in L^2 (\R^3)$ and $\Psi \in H^2 (\R^3)$.
Since $V_1{\vecLambda}$ and
$V_1|{\vecLambda}|^2$ are in $L^2(\R^3)$, estimate (\ref{estimate-inf})
with $\Psi=(S_0-\z)^{-1}\phi$,  $\phi\in{L}^2(\R^3)$, suggests to write both
operators ${V}_1(S_0-\z)^{-1}{\vecLambda}$ and $V_1(S_0-\z)^{-1}
|{\vecLambda}|^2$
with
the resolvent $(S_0-\z)^{-1}$ shifted to the right side by using commutator properties.
Thus we consider the following expressions:
\bea\label{eqtwo-terms}
{V}_1(S_0-\z)^{-1}{\vecLambda} &=&  V_1{\vecLambda}(S_0-\z)^{-1}  \nonumber \\ &&+
\frac{\imath}{2} V_1(S_0-\z)^{-1}(-\hat{p}_2,\hat{p}_1,0)(S_0-\z)^{-1}  \label{term-1}
\\
V_1(S_0-\z)^{-1}
|{\vecLambda}|^2 &=& V_1|{\vecLambda}|^2(S_0-\z)^{-1} \nonumber \\
&&+
\frac{\imath}{2}V_1\left[x_1(S_0-\z)^{-2}\hat{p}_1 + x_2(S_0-\z)^{-2}\hat{p}_2\right]  \nonumber \\ &&-
\frac{1}{2}V_1(S_0-\z)^{-3}\left(\hat{p}_1^2+\hat{p}_2^2\right) + \frac{1}{2}V_1(S_0-\z)^{-2} \label{term-2}
\eea

\noindent
4. We next estimate the right sides of \eqref{term-1} and \eqref{term-2}. We first note that for $q=0,1,2$, there exist a constant C such that, for $j=1,2,3$,
$\| V_1\Lambda_j^{q}\|_2\leq{C} R^{(2q+1)/2}$, with
$\Lambda_j$ denoting the $j^{\rm th}$-component of the operator ${\vecLambda}$.
Thus, from estimates \eqref{estimate-inf}-\eqref{estimate-inf2}, we have for $\beta>0$
\bea
\|  V_1\Lambda_j ^q(S_0-\z)^{-1}\| \leq {C} R^{\frac{(2q+1)}{2}}
\left[\frac{1}{\beta^{1/2}} \|\Delta(S_0-\z)^{-1}\|
 + \beta^{3/2} \|(S_0-\z)^{-1}\|\right],
 \nonumber \\
   \;\;  j=1,2,3 \;\;   {\rm and } \;\;     q=0,1,2.  \;\; \label{est-bd}
\eea
Since $\|\frac{1}{2}\Delta(S_0-\z)^{-1}\|\leq{\rm max}\{ s/\left| s-\z \right| \; | \; s\in[0,\infty)\}
\leq \eta(\z)$, with $\eta(\z)=1$, if  $\Re(\z)\leq{0}$ and   $\eta(\z)=|\z|/|\Im(\z)|$ if $\Re(\z)>{0}$, we obtain
from (\ref{eq-bd1})-(\ref{eq-bd2}) (with $\lambda=0$) together with (\ref{est-bd}):
 \bea
\|{V}_1(S_0-\z)^{-1}{\vecLambda}\| &\leq& C\left[R^{3/2} + \frac{R^{1/2}}{\sqrt{G(\z)}}
\right]\left[\frac{2\eta(\z)}{\beta^{1/2}}  + \frac{\beta^{3/2}}{d(\z)}\right],  \phantom{xxxxxxxxxxxx} \label{bound1}   \\
\|V_1(S_0-\z)^{-1}|{\vecLambda}|^2\|  &\leq& C\left[R^{5/2} + \frac{R^{3/2}}{\sqrt{G(\z)}}
 + \frac{R^{1/2}}{G(\z)} + \frac{R^{1/2}}{d(\z)}\right]\nonumber \\
  &&\cdot\left[\frac{2\eta(\z)}{\beta^{1/2}}  + \frac{\beta^{3/2}}{d(\z)}\right]. \label{bound2}
\eea
 Therefore, we get from (\ref{eqV_1})-(\ref{eq-bd2}), and (\ref{bound1})-(\ref{bound2})
that for $\beta>0$ :
\bea
\|V_1\left[(S_0 (\lambda) - \z)^{-1} - (S_0 - \z)^{-1}\right]\|
\leq C \left[
\lambda\left(R^{3/2} + \frac{R^{1/2}}{\sqrt{G(\z)}}\right)
\frac{1}{\sqrt{G(\z)}} \right.
\nonumber \\  \left. +
\lambda^2\left(R^{5/2} + \frac{R^{3/2}}{\sqrt{G(\z)}}
+ \frac{R^{1/2}}{G(\z)} + \frac{R^{1/2}}{d(\z)}\right)\frac{1}{d(\z)}\right]
\left[\frac{2\eta(\z)}{\beta^{1/2}}  + \frac{\beta^{3/2}}{d(\z)}\right].  \label{estV_1}
\eea
We conclude the proof of part (1) of Lemma \ref{lemma:convergence-norm1} from \eqref{estV_1} and (\ref{estV_2}) by first taking the limit $\lambda\rightarrow{0}$ with $R$ fixed and then letting $R\rightarrow\infty$.

\noindent
5. For the proof of part (2), let us consider $\lambda=\lambda(h=1/(N+1))$, with $N \in \N$,  in
(\ref{estV_1})  and (\ref{estV_2}).
Let us take $R=N^{\gamma}/d(\z)$, $\gamma>0$,  in equation (\ref{estV_2})
in order to have
\beq\label{estV2}
\|V_2\left[(S_0 (\lambda(h)) - \z)^{-1} - (S_0 - \z)^{-1}\right]\| = O(N^{-\gamma})
\eeq
We will choose $\gamma$ below.
Since $|\z-E_N|=O(N^{-3})$  then $d(\z)=|\z|=O(N^{-2})$,
and $\Re(\z)<{0}$  for $N$ sufficiently large, which implies $\eta(\z)=1$.
Whence $R=O(N^{2+\gamma})$. Thus we can estimate the two terms appearing in the first factor in square brackets in (\ref{estV_1}) by using both estimates $\lambda=h^3\epsilon(h)B=O(N^{-3-q})$
and $1/G(\z)\leq1/|\z|$:
\bea
\lambda\left(R^{3/2} + \frac{R^{1/2}}{\sqrt{G(\z)}}\right)\frac{1}{\sqrt{G(\z)}} &=&
O\left(N^{1-q+\frac{3\gamma}{2}}\right),  \\
\lambda^2\left(R^{5/2} + \frac{R^{3/2}}{\sqrt{G(\z)}}
 + \frac{R^{1/2}}{G(\z)} + \frac{R^{1/2}}{d(\z)}\right)\frac{1}{d(\z)} &=&
 O\left(N^{1-2q+\frac{5\gamma}{2}}\right) .
\eea
Now we replace
the factor
$\left[ \frac{2\eta(\z)}{\beta^{1/2}}  + \frac{\beta^{3/2}}{d(\z)} \right]$ appearing in  (\ref{estV_1})  by  its minimum value $C(d(\z))^{-1/4}=O(N^{1/2})$.   Thus we have
\beq\label{eq-est-3term}
\|V\left[(S_0 (\lambda) - \z)^{-1} - (S_0 - \z)^{-1}\right]\| =
O\left(N^{\frac{3}{2}-q+\frac{3\gamma}{2}}\right) + O(N^{\frac{3}{2}-2q+\frac{5\gamma}{2}})
+  O(N^{-\gamma}).
\eeq

\noindent
6. If we take $q>3/2$ given then there exist $0<\gamma<\frac{2}{3}q-1$ such that
both exponents $\frac{3}{2}-2q+\frac{5\gamma}{2}$  and   $\frac{3}{2}-q+\frac{3\gamma}{2}$ are negative. Moreover, since $\gamma<\frac{2}{3}q-1<q$  then $\frac{3}{2}-2q+\frac{5\gamma}{2} < \frac{3}{2}-q+\frac{3\gamma}{2}$,
which implies that in the regime $q>3/2$ and $0<\gamma<\frac{2}{3}q-1$
the contribution from the linear term in $\lambda$ dominates the quadratic one in equation (\ref{estV_1}).  From equations (\ref{eq-est-3term}) and (\ref{estV2})
\beq\label{almost-fin-est}
\|V\left[(S_0 (\lambda(h)) - \z)^{-1} - (S_0 - \z)^{-1}\right]\| =   O\left(N^{3/2-q+\frac{3\gamma}{2}}\right)  +  O\left(N^{-\gamma}\right)
\eeq
Let us write the right hand side of equation (\ref{almost-fin-est}) as $O(N^{-E_q(\gamma)})$ with
$E_q(\gamma)=\min\{q-3/2-\frac{3\gamma}{2}, \gamma\}$.  Working in the regime specified above, we actually have that the maximum value of $E_q(\gamma)$ is
$E_q(\frac{2}{5}q-\frac{3}{5})=\frac{2}{5}q-\frac{3}{5}$.  Hence we finally have
\beq
\|V\left[(S_0 (\lambda(h)) - \z)^{-1} - (S_0 - \z)^{-1}\right]\|  =
O\left(N^{-\frac{2q-3}{5}}\right) .
\eeq
This completes the proof of part (2).
\end{proof}

We prove some resolvent estimates that are needed in the proof of the stability theorem in the next section.
For an eigenvalue $E_N$ of $S_V$, let $\Gamma_N$ be a circle with center $E_N$ and radius $r_N=cN^{-3}$, with $c$ a suitable constant independent of $N$ in such a way that $r_N$ is smaller than half the distance to the nearest eigenvalue, which is $\mathcal{O}(N^{-3})$.

\begin{lemma}\label{lemma:resolvent-est1}
Uniformly for all $\z \in \Gamma_N$, we have
\begin{enumerate}
\item $\| V (S_0 - \z)^{-1} \| = \mathcal{O}(N)$

\item $\| V (S_V - \z)^{-1} \| = \mathcal{O}(N^3)$

\end{enumerate}
\end{lemma}

\begin{proof}
1. for $R > 0$, we decompose $V= V \chi_R + V (1 - \chi_R) = V_1+V_2$, as in the proof of Lemma \ref{lemma:convergence-norm1},
with $V_1 \in L^2 (\R^3)$, $V_2 \in L^\infty (\R^3)$,
 $\|V_1\|_{2} = 2 \sqrt{\pi} R^{1/2}$ and $\|V_2\|_\infty = {1/R}$.
Using estimate \eqref{estimate-inf}, we have for all $\beta>0$
\bea
&&\|   V\left(S_0-\z\right)^{-1}  \| \leq \frac{2C\|V_1\|_2}{\beta^{1/2}}  \nonumber \\
&&  \phantom{xxxxxx}+
\left[  \frac{2C\|V_1\|_2|\z|}{\beta^{1/2}} +  C\|V_1\|_2\beta^{3/2} + \|V_2\|_\infty\right]
\|\left(S_0-\z\right)^{-1} \|  \nonumber \\
 &&\phantom{xxxxx}\leq \frac{2CR^{1/2}}{\beta^{1/2}}  +
\left[  \frac{2CR^{1/2}|\z|}{\beta^{1/2}} +  CR^{1/2}\beta^{3/2} + \frac{1}{R}\right]
\|\left(S_0-\z\right)^{-1}\| \phantom{xxx} \nonumber \\
&&\phantom{xxxxx}\leq 2CN^\mu + \left[ 2CN^\mu|\z| + CN^\mu\beta^2 + \frac{1}{N^{2\mu}\beta}\right] \|\left(S_0-\z\right)^{-1}\|,
\eea
where we have written $R= \beta N^{2\mu}$, with $\mu\in\R$. Then we optimize
the function $g(\beta)=CN^\mu\beta^2 + \frac{1}{N^{2\mu}\beta}$ by its minimium value,
with $N^\mu$ fixed,
and use the estimates  $|\z|=O(N^{-2})$ and
$\|\left(S_0-\z\right)^{-1}\|=O(N^{2})$, in order to get the estimate
$\|   V\left(S_0-\z\right)^{-1}  \| = O(N^\mu) + O(N^{2-\mu})$, which is optimal when $\mu=1$.

\noindent
2. As for part (2), from the estimate (\ref{estimate-inf}) with $C > 0$ as there, we obtain for all $\beta>{0}$:
\bea\label{bound-inv}
\| V\left(S_V-\z\right)^{-1} \| & \leq & \gamma(V_1, \beta) 
\left[\frac{1}{\beta^{1/2}} +
\left(\frac{|z|}{\beta^{1/2}} + \beta^{3/2} \right)\|\left(S_V-z\right)^{-1}\| \right]
 \nonumber \\
 & & + \gamma (V_1, \beta)\|V_2\|_{\infty}  \|\left(S_V-z\right)^{-1}\| ,
\eea
as long as $\gamma (V_1, \beta) := 1-\frac{C\|V_1\|_{2}}{\beta^{1/2}}$ is strictly positive.
In order to optimize the upper bound in equation (\ref{bound-inv}),
we take $R= \frac{\beta}{(4 \sqrt{\pi} C)^2 N^{\alpha}}$, with $\alpha\geq{0}$. 
With this choice, the factor satisfies $\gamma (V_1, \beta)= 1 - (2 N^{\alpha/2})^{-1} \geq 1/2$.
Optimizing with respect to $\beta$, and using $\|\left(S_V-z\right)^{-1}\|=O(N^{3})$, we conclude that
$\| V\left(S_V-\z\right)^{-1}\| =O(N^{3+\alpha/2})$, which is optimal when $\alpha = 0$.
\end{proof}

\subsection{Eigenvalue stability.}\label{subsec:ev-clusters1}

We next prove the main result on eigenvalue stability by following  \cite{AhS1} adapted to our setting.

Let us first recall from Kato \cite[chapter VIII, section 1, part 4]{kato} that an isolated eigenvalue $E_0$ of a closed operator $T_0$
with finite multiplicity $N_0$ is stable with respect to a family of closed perturbations $T_n$ if
\begin{enumerate}
\item There exists an $\epsilon > 0$, so that any $\z$ with $0 < | \z - E_0| < \epsilon$ is not in the spectrum of $T_n$,
for all $n$ large (depending on $\z$), and for such a $\z$, we have $(T_n - \z)^{-1} \rightarrow (T - \z)^{-1}$ strongly;
\item The total multiplicity of the eigenvalues of $T_n$ in a neighborhood of $E_0$  given by
$\{ \z ~|~ 0 \leq | \z - E_0| < \mu \}$, with $0<\mu<\epsilon$,
is precisely $N_0$ for all $n$ large.
\end{enumerate}
It is proven in \cite[chapter VIII, section 1, part 4, Lemma 1.24]{kato} that if $E_0$ is stable with respect to the family $T_n$, and all the operators are self-adjoint, then, in fact, the spectral projectors converge in norm. That is, by part 1 of the definition and $0<\mu<\epsilon$,
the
contour $\Gamma_{E_0,\mu} = \{ \z ~|~ |\z - E_0| = \mu \}$ is in the resolvent sets of $T_n$, for all $n$ large. We can then define the projectors
\beq\label{eq:proj0}
P_n = - \frac{1}{2 \pi i} \int_{\Gamma_{E_0,\mu}} (T_n - \z)^{-1} ~d\z ,
\eeq
and, similarly, we define the projector $P$ for $T$ using the same contour $\Gamma_{E_0,\mu}$. The self-adjointness of $T_n$ and $T$ imply that these are orthogonal projectors. By part 1, we have $P_n \rightarrow P$ strongly, and by part 2, $\dim {\rm Ran} P_n = N_0$, for all $n$ large.
Under these conditions, Kato proves that $\| P_n - P \| \rightarrow 0$, as $n \rightarrow \infty$.

We are interested in studying a situation that is not exactly the one described in the above definition by Kato but very much in the same spirit. Namely, given $N\in\N$, let us consider the operator $S_V(\lambda(h=1/(N+1),B)) =  S_V  + W(\lambda(h=1/(N+1),B))$ (see equation
(\ref{eq:rescale1})).   Here we have the  fixed (N-independent)
operator  $S_V= -  \frac{1}{2}\Delta - \frac{1}{|x|} $ plus a perturbation  $W(\lambda(h=1/(N+1),B))=\frac{\left(\lambda(h=1/(N+1),B)\right)^2}{8} (x_1^2 + x_2^2) - \frac{\lambda(h=1/(N+1),B)}{2} L_3$ indexed by $N$.
We  want to look at the eigenvalues
$E_N = -1/2(N+1)^2$
of the hydrogen atom Hamiltonian $S_V$ (remember that each eigenvalue $E_N$ has
multiplicity $d_N=(N+1)^2$).  We want to show that, for $N$ sufficiently large,  the spectrum of $S_V(\lambda(h=1/(N+1),B))$ inside some neighborhood around $E_N$ consists of
a cluster ${\mathcal C}_N$ of $d_N$ eigenvalues (including multiplicity).
The size $r_N$ of such a neighborhood should decrease with $N$.  Notice that in this situation, both the eigenvalue $E_N$ and $r_N$
change with $N$, which is not the case in  the definition of stability by Kato where
  the eigenvalue $E_0$ is fixed and the size $\epsilon$ of the neighborhood around $E_0$ can be kept   fixed as well.   Since the cluster
${\mathcal C}_N$ can be thought  of as splitting off of the unperturbed eigenvalue $E_N$ into several eigenvalues of total multiplicity $d_N$, we will refer to  the existence of  ${\mathcal C}_N$  as a stability property.

In order to show the existence of the cluster ${\mathcal C}_N$,  we first regard $E_N$ as an element of the discrete spectrum of $S_V$ and  notice that the distance $\rho(N)$
between $E_N$ and its nearest neighbors is
$O(N^{-3})$.   Thus we want to consider a circle $\Gamma_N$ with center $E_N$ and radius $r_N=cN^{-3}$, with $c$ a suitable constant independent of $N$ in such a way that $r_N$ is smaller than $\rho(N)/2$.
Then we have the following:

\begin{theorem}[Stability theorem]\label{thm:stability1}
Given $B\geq{0}$ and suppose that the constant $q$ in part 2 of Lemma \ref{lemma:convergence-norm1} satisfies $q>9$.
The following spectral projectors are well-defined for $N$ sufficiently large:
\bea
P_N&=&-\frac{1}{2\pi\imath}\int_{\Gamma_N} \left(S_V(\lambda(h=1/(N+1),B))-\z\right)^{-1}
d\z, \label{projQ_N} \\
\Pi_N&=&-\frac{1}{2\pi\imath}\int_{\Gamma_N} \left(S_V-\z\right)^{-1}d\z  \label{projP_N} .
\eea
Moreover, these projectors are orthogonal and satisfy
\beq
\|P_N-\Pi_N\| = O(N^{-\frac{2q-33}{5}}).
\eeq
For $q>33/2$, the difference of the orthogonal projectors $P_N - \Pi_N$ converges in norm to zero.
Consequently, the spectrum of
$S_V(\lambda(h=1/(N+1),B))$ inside the circle $\Gamma_N$
consist of a cluster $\mathcal{C}_N$ of eigenvalues with total multiplicity $d_N$ provided $N$ is sufficiently large.
\end{theorem}

\begin{proof}
We follow Avron, Herbst and Simon \cite{AhS1} in obtaining specific upper bounds on the difference $P_N - \Pi_N$.
For the purpose of the proof of Theorem \ref{thm:stability1}, we will only write $\lambda$ to actually specify $\lambda(h=1/(N+1),B)$, assuming $B>0$.

\noindent
1. We first establish the existence of the resolvent operator
 $\left(S_V(\lambda)-\z\right)^{-1}$ for $\z\in\Gamma_N$.
 Since both resolvents
 $\left(S_V-\z\right)^{-1}$ and $\left(S_0-\z\right)^{-1}$ exist for $\z\in\Gamma_N$, the equality
 $S_V-\z = \left[I + V\left(S_0-\z\right)^{-1}\right](S_0-\z)$ implies that the operator
 $I + V\left(S_0-\z\right)^{-1}$ is invertible.  Thus for $\lambda$ small, we expect from Lemma  \ref{lemma:convergence-norm1} that the operator  $I + V\left(S_0(\lambda)-\z\right)^{-1}$ is invertible as well.
 This, and the invertibility of $S_0(\lambda)-\z$ for $\z\in\Gamma_N$, imply that $S_V(\lambda)-\z$ is invertible for $\z\in\Gamma_N$ since
 $$
 S_V(\lambda)-\z = S_0(\lambda)+V-\z = \left[I + V\left(S_0(\lambda)-\z\right)^{-1}\right](S_0(\lambda)-\z).
 $$

\noindent
2. To establish the invertibility $I + V\left(S_0(\lambda)-\z\right)^{-1}$, we write
\bea\label{invert}
I + V\left(S_0(\lambda)-\z\right)^{-1} &=& \left\{  I + V\left[\left(S_0(\lambda)-\z\right)^{-1}-\left(S_0-\z\right)^{-1}\right]\right. \nonumber \\
&& \left.  \cdot\left[I + V\left(S_0-\z\right)^{-1}\right]^{-1}
 \right\}\left[ I + V\left(S_0-\z\right)^{-1} \right] .
\eea
Because of the estimate in part (2) of Lemma \ref{lemma:convergence-norm1}, we need to estimate
\beq\label{eq:invert3}
\left\| \left[I + V\left(S_0-\z\right)^{-1}\right]^{-1} \right\|
\eeq
in order to have
\beq\label{invert-criterio}
\left\|  V\left[\left(S_0(\lambda)-\z\right)^{-1}-\left(S_0-\z\right)^{-1}\right] \left[I + V\left(S_0-\z\right)^{-1}\right]^{-1} \right\| < 1 ,
\eeq
for $\z \in \Gamma_N$, which would imply  the invertibility of $I + V\left(S_0(\lambda)-\z\right)^{-1}$ for $N$ large.

We write
\bea\left[I + V\left(S_0-\z\right)^{-1}\right]^{-1} &=&
\left[\left(S_0-\z+V\right)\left(S_0-\z\right)^{-1}\right]^{-1}  \nonumber \\ &=&
(S_V-\z-V)(S_V-\z)^{-1} =
I-V\left(S_V-\z\right)^{-1} .\;\:\;\;\;\;\;\;
\eea
From part (2) of Lemma \ref{lemma:resolvent-est1}, we have
%
\beq\label{est-inv}
\left\| \left[I + V\left(S_0-\z\right)^{-1}\right]^{-1}  \right\|= O(N^3).
\eeq
Using Lemma \ref{lemma:convergence-norm1} and (\ref{est-inv}), we obtain
\beq\label{est-inv-II}
\left \|  V\left[\left(S_0(\lambda)-\z\right)^{-1}-\left(S_0-\z\right)^{-1}\right] \left[I + V\left(S_0-\z\right)^{-1}\right]^{-1} \right\| = O(N^{-\frac{2q-18}{5}}) .
\eeq
Thus, in order to satisfy condition (\ref{invert-criterio}), and then to have the existence of the resolvent $(S_V(\lambda)-\z)^{-1}$ with $|\z-E_N|=r_N $, we need to take $q>9$.

\noindent
3. We next apply these estimates to bound the difference of the projectors from above. As the resolvents have been shown to exist on the contour $\Gamma_N$, we  have
\beq\label{dif-res}
P_N-\Pi_N = -\frac{1}{2\pi\imath}\int_{\Gamma_N}\left[\left(S_V(\lambda)-\z\right)^{-1} -
\left(S_V-\z\right)^{-1}\right]d\z.
\eeq
We know  the  operator $S_V(\lambda)$ does not converge to $S_V$
in the norm resolvent sense as $\lambda\rightarrow{0}$ \cite{AhS1}.   The key ideas necessary to
obtain norm convergence of the left hand side of (\ref{dif-res})
consist in {\bf (i)} inserting in the right hand side of equation (\ref{dif-res}) the integrals
$ \frac{1}{2\pi\imath}\int_{\Gamma_N} \left(S_0(\lambda)-\z)\right)^{-1} d\z$
 and $ -\frac{1}{2\pi\imath}\int_{\Gamma_N} \left(S_0-\z)\right)^{-1} d\z$ which are zero due to the analyticity of the resolvents $\left(S_0(\lambda)-\z\right)^{-1}$ and
 $\left(S_0-\z\right)^{-1}$
 on $\C \backslash [0,\infty)$, and then {\bf (ii)} using the convergence of
 $\left(S_V(\lambda)-\z\right)^{-1} - \left(S_0(\lambda)-\z\right)^{-1}$ to
 $\left(S_V-\z\right)^{-1}-\left(S_0-\z\right)^{-1}$ as $\lambda\rightarrow{0}$. In order to estimate this last convergence,  we write
 \bea
 &&\left\| \left[ \left(S_V(\lambda)-\z\right)^{-1} - \left(S_0(\lambda)-\z\right)^{-1}\right] -
 \left[\left(S_V-\z\right)^{-1}-\left(S_0-\z\right)^{-1}\right] \right\|  \phantom{xxxxx}\nonumber \\
 &&  \phantom{xxxx} =\left\|\  \left(S_0(\lambda)-\z\right)^{-1}V\left(S_V(\lambda)-\z\right)^{-1} -
 \left(S_0-\z\right)^{-1}V \left(S_V-\z\right)^{-1}  \right\|
  \nonumber \\
&&  \phantom{xxxx}  = \left\|\   \left(S_0(\lambda)-\z\right)^{-1}\left[V\left(S_V(\lambda)-\z\right)^{-1} -
 V\left(S_V-\z\right)^{-1} \right] \right.\nonumber \\
 &&  \phantom{xxxxxx} \left.+ \left[ \left(S_0(\lambda)-\z\right)^{-1}V-\left(S_0-\z\right)^{-1}V\right]
\left(S_V-\z\right)^{-1}
   \right\| \nonumber \\
    &&  \phantom{xxxxxx}  +
      \left\|\    \left(S_0(\lambda)-\z\right)^{-1}V-\left(S_0-\z\right)^{-1}V  \right\|
       \left\|\   \left(S_V-\z\right)^{-1}   \right\|.   \label{estdifres}
    \eea
In light of resolvent estimates \eqref{eq:basic-norm-bounds1}, and the estimate \eqref{eq:norm2},
we can control the norm of the difference of the resolvents in \eqref{estdifres} provided we can estimate
\beq\label{eq:resolvent-diff1}
\left\|\  V\left(S_V(\lambda)-\z\right)^{-1} -
 V\left(S_V-\z\right)^{-1}   \right\|.
 \eeq

\noindent
4. To obtain an upper bound on the norm in \eqref{eq:resolvent-diff1},
we write
 \bea 
 V\left(S_V(\lambda)-\z\right)^{-1} &=& V\left(S_0(\lambda)-\z\right)^{-1}
 \left[I + V\left(S_0(\lambda)-\z\right)^{-1}\right]^{-1} \label{vinvert1} \\
  V\left(S_V-\z\right)^{-1} &=& V\left(S_0-\z\right)^{-1}
 \left[I + V\left(S_0-\z\right)^{-1}\right]^{-1} . \label{vinvert2}
 \eea
The invertibility of $I + V\left(S_0-\z\right)^{-1}$, appearing in \eqref{vinvert2}, was established in \eqref{est-inv}. We prove that $I + V\left(S_0 (\lambda)-\z\right)^{-1}$ in \eqref{vinvert1} is invertible as follows. 
We write this factor as
\bea\label{eq:vinvert3}
\left[I + V\left(S_0(\lambda)-\z\right)^{-1}\right] & = & \left[I + V\left(S_0-\z\right)^{-1}\right]\left[I + \left( I + V\left(S_0-\z\right)^{-1}\right)^{-1} \right. \nonumber \\
   & & \left. \times \left\{ V\left(S_0(\lambda)-\z\right)^{-1} -  V\left(S_0-\z\right)^{-1}
 \right\}\right] .
\eea
It follows from \eqref{est-inv} and part (2) of Lemma \ref{lemma:convergence-norm1} that
\beq\label{eq:vinvert4}
\| \left( I + V\left(S_0-\z\right)^{-1}\right)^{-1} \left\{ V\left(S_0(\lambda)-\z\right)^{-1} -  V\left(S_0-\z\right)^{-1}
 \right\} \| = \mathcal{O}(N^{- \left( \frac{2q - 18}{5} \right)}) ,
 \eeq
so each factor in square brackets on the right side of \eqref{eq:vinvert3} is invertible. Furthermore, it follows from
\eqref{eq:vinvert3}--\eqref{eq:vinvert4} that
\beq\label{eq:vinvert5}
\left\| \left[I + V\left(S_0(\lambda)-\z\right)^{-1}\right]^{-1} \right\| = \mathcal{O}(N^3).
\eeq
Returning to \eqref{vinvert1}--\eqref{vinvert2}, we have
\bea\label{eq:resolvent-difference2}
\lefteqn{ \| V(S_V(\lambda) - z)^{-1} - V(S_V -z )^{-1} \|} & & \nonumber \\
  & \leq & \| V(S_0(\lambda) - z)^{-1} - V(S_0 -z )^{-1} \|
   \| [ 1 + V(S_0(\lambda) -z)^{-1} ]^{-1}\| \nonumber \\
 & & + \| V(S_0 - z)^{-1} \| \| [ 1 + V(S_0(\lambda) -z)^{-1} ]^{-1} - [ 1 + V(S_0(\lambda) -z)^{-1} ]^{-1} \|. \nonumber \\
 \eea
The first term on the right in \eqref{eq:resolvent-difference2} is bounded from part (2) of Lemma \ref{lemma:convergence-norm1} and \eqref{eq:vinvert5}. As for the second term on the right,
we note that from equations (\ref{est-inv})-(\ref{est-inv-II}) we obtain
 \beq\label{invaprox}
 \left\| \left[I + V\left(S_0(\lambda)-\z\right)^{-1}\right]^{-1} - \left[I + V\left(S_0-\z\right)^{-1}\right]^{-1} \right\|
  = O\left(N^{-\frac{2q-33}{5}}\right).
 \eeq
and $\|V(S_0-\z)^{-1}\| = O(N)$ from part (1) of Lemma \ref{lemma:resolvent-est1}. Consequently, we have the bound
\beq\label{Vdifest}
\| V\left(S_V(\lambda)-\z\right)^{-1} -  V\left(S_V-\z\right)^{-1} \| = O(N^{-\frac{2q-38}{5}}) .
\eeq

\noindent
5. We conclude the convergence of the projectors as follows. From Lemma \ref{lemma:convergence-norm1},  equation (\ref{Vdifest}), and the norm estimates
$\left\|\left(S_0(\lambda)-\z\right)^{-1}\right\|=O(N^2)$ and
$\left\|\ \left(S_V-\z\right)^{-1} \right\| = O(N^3)$,
we finally get the estimate for the difference of the spectral projectors
\bea
\|  P_N-\Pi_N    \|  &\leq& \frac{1}{2\pi} \int_{\Gamma_N}\left\| \left[ \left(S_V(\lambda)-\z\right)^{-1} - \left(S_0(\lambda)-\z\right)^{-1}\right]
\right. \nonumber \\
&&   \left.    \phantom{xxx} -
 \left[\left(S_V-\z\right)^{-1}-\left(S_0-\z\right)^{-1}\right] \right\| |d\z|  = O(N^{-\frac{2q-33}{5}}).  \phantom{xxxx}
\eea
Thus taking $q>33/2$ we have, for $N$ sufficiently large,   $\|  P_N-\Pi_N    \|  < 1$
 and then that the dimension of the range of both projectors $P_N$ and $\Pi_N$ is the same (see reference \cite{kato}) which in turn implies the existence of the cluster ${\mathcal C}_N$.
\end{proof}


\section{A semiclassical trace identity for Zeeman eigenvalue clusters.}\label{sec:trace1}

From Theorem \ref{thm:stability1}, we know that for $q>33/2$ the size of the eigenvalue
cluster ${\mathcal C}_N$  around
$E_N$ is no larger than $r_N=O(N^{-3})$.
We need to get a better estimate on  the size of ${\mathcal C}_N$ in order to scale the shifts of eigenvalues within ${\mathcal C}_N$.  Let us first consider the case of the eigenvalue cluster formed by the perturbation
$S_V(\lambda)$ of $S_V$, where $\lambda \geq 0$ is the magnetic field strength independent of any other parameters.
In this case, Theorem 5.6 in reference  \cite{AhS1}, when applied to the Coulomb potential,
shows that if we take a {\it fixed} eigenvalue $E_M$ of the scaled hydrogen atom Hamiltonian $S_V$ with multiplicity $d_M= (M+1)^2$, then we have for $\lambda$ sufficiently small that the operator $S_V(\lambda)$ has a cluster of eigenvalues $E_{M, j}(\lambda)$, $j=1,\ldots,d_M$
around $E_M$ inside a small but {\it fixed} circle with center $E_M$.  The eigenvalues in the cluster can be written in the following way:
Let $m = -(M-1), \ldots, M-1$ be the eigenvalues of $\Pi_M L_3 \Pi_M$, where $\Pi_M$ projects onto the eigenspace of $S_V$ and eigenvalue $E_M$. Then, for
a given $j$, there exist  an index $m$ so that
\beq\label{eq:cluster-size2}
E_{M, j}(\lambda) = E_M  - \lambda m  + O (\lambda^2).
\eeq
This indicates that the size of the cluster of eigenvalues around $E_M$ is $O(\lambda)$, with $M$ fixed.

In our setting, the parameter $N$ controls both the strength of the effective magnetic field  $\lambda=\lambda(h,B)=h^3\epsilon(h)B$, $h=1/(N+1)$, and the radius $r_N$ of the circles
${\mathcal C}_N$ where our clusters of eigenvalues are well defined.  Thus  neither the center $E_N$
nor the radius $r_N$   stay fixed as it is the case of the mentioned result of reference \cite{AhS1}.  So we need to get estimates considering that fact.   We begin with norm  estimates of the operators $\Pi_NL_3\Pi_N$ and $ \Pi_N(x^2+y^2)\Pi_N$ :

\begin{lemma}\label{norm-estimates-N}
Let $\Pi_N$ be the projector to the eigenspace associated to the eigenvalue $E_N$ of the scaled hydrogen atom Hamiltonian $S_V$ as above.  Then for $N\rightarrow\infty$  and $k\in{\N}$ fixed we have
\bea
 \Pi_NL_3\Pi_N&=&O(N), \label{est-ang-mom} \\
 \Pi_N(x^2+y^2)^k\Pi_N&=&O(N^{4k}).  \label{est-quad-pot}
\eea
\end{lemma}

\begin{remark}
The physical intuition for equation (\ref{est-quad-pot}) comes from the Kepler  problem.   In that case,  the maximum apogee distance $r_{\rm max}(E)$ for an
 orbit in configuration space of negative energy $E$ is $1/|E|$ (including collision and non-collision orbits). So,  $r_{\rm max}(E=-1/2)=2$ and   $r_{\rm max}(E=-1/(2N^2))=2N^2$. Thus we expect, semiclassically speaking , that  for a Kepler orbit in configuration space,  $x^2+y^2=O(N^4)$.
This property is implemented via coherent states
$\Psi_{\alpha,N}$ with $N$ large.
\end{remark}

Before presenting the proof of Lemma \ref{norm-estimates-N}, we briefly recall some facts about the coherent states
$\Psi_{\alpha,N}$,
complete details are presented in Appendix 2, section \ref{sec:app2-coherentstates1}.
The index $\alpha\in\C^4$ is an element of
${\mathcal A}=\{\alpha\in{\C^4} \; | \; |\Re(\alpha)| = |\Im(\alpha)| =1,
\;  \Re(\alpha)\cdot \Im(\alpha)=0  \}$.
The set ${\mathcal A}$
can be thought of as the unit cotangent bundle
$T_1^*\Sp^3$
of the 3-sphere $\Sp^3$
with $\Re\alpha\in{\Sp}^3$ and
$\Im\alpha$ an element of the cotangent space to $\Sp^3$ at the point $\Re\alpha$.   The inverse of the Moser map  (see Appendix 1, section \ref{sec:moser-map}) then relates $T_1^*\Sp_o^3$ (the unit cotangent bundle $T_1^*\Sp^3$ minus the north pole) with the energy surface $\Sigma(-1/2)$ of the Kepler problem.
The states  $\Psi_{\alpha,N}$
belong to the range of $\Pi_N$ and provide a resolution of the identity giving an expression for the projector $\Pi_N$ in terms of them,  see
equation (\ref{eq:proj-coherentstates1}) in Appendix 2, section \ref{sec:app2-coherentstates1}.

\begin{proof}
1. Equation (\ref{est-ang-mom}) comes from the well known fact that the eigenvalues of $L_3$ restricted to the range of $\Pi_N$ are $-N, \ldots, N$.

\noindent
2. We use the coherent states $\Psi_{\alpha,N}$, for $\alpha \in \mathcal{A}$
in the proof of \eqref{est-quad-pot}.
In \cite{thomas-villegas1}, it is shown that the dilated coherent state $D_{(N+1)^2}\Psi_{\alpha,N} $ has a fast decay  outside of a ball of radius $r_0>2$.  Specifically,  we have the following result shown in \cite[4.19]{thomas-villegas1}:

\begin{lemma}\label{decay-lemma-proj}
 Let $\tilde{V}:\R^3\rightarrow\R$ be a polynomially bounded continuous function.  Then for
 $r_0>2$ we have
 \beq
\Pi_ND_{(N+1)^{-2}}\tilde{V}D_{(N+1)^2}\Pi_N =  \Pi_ND_{(N+1)^{-2}}\tilde{V}\chi_{|x|\leq{r_0}}D_{(N+1)^2}\Pi_N
+ O(N^{-\infty})
 \eeq
 where $\chi_{|x|\leq{r_0}}$ is the characteristic function of the ball $|x|\leq{r_0}$ and the dilation operator $D_\alpha$ is defined in \eqref{eq:dilation0}.
 \end{lemma}
\noindent
The proof of  equation
(\ref{est-quad-pot}) is then a consequence of Lemma \ref{decay-lemma-proj} and the following equalities:
 \bea
 &&D_{(N+1)^{-2}}(x^2+y^2)^kD_{(N+1)^2}=\frac{1}{(N+1)^{4k}}(x^2+y^2)^k \nonumber \\
 &&  \Pi_N(x^2+y^2)^k\Pi_N=(N+1)^{4k}\Pi_ND_{(N+1)^{-2}}(x^2+y^2)^{k}D_{(N+1)^2}\Pi_N.
 \eea
Thus we finally have $\|\Pi_ND_{(N+1)^{-2}}(x^2+y^2)^{k}D_{(N+1)^2}\Pi_N\|=O(1)$ and then
the proof of (\ref{est-quad-pot}).
\end{proof}

\begin{remark}  We can actually say more about the coherent states
 $D_{(N+1)^2}\Psi_{\alpha,N} $ in terms of concentration.  It can be shown that for  ${\alpha}\in{\mathcal A}$ given, the state $D_{(N+1)^2}\Psi_{\alpha,N} $ is highly concentrated
along the classical orbit in configuration space associated to $\alpha$ by the inverse of the Moser map and the classical flow of the Kepler problem on the energy surface $\Sigma(-1/2)$. See references
\cite{thomas-villegas1} and
\cite{villegas} for details. \end{remark}


Let us denote by $\tilde{E}_{N,j}$, $j=1,\ldots,d_N$,
the eigenvalues of $S_V(\lambda)$ inside the circle ${\Gamma}_N$  (this notion is well defined for $N$ sufficiently large). Now we  consider the  eigenvalue shifts $\tilde\nu_{N,j}=\tilde{E}_{N,j}-E_N$ thinking of  them as the eigenvalues of the operator $P_N(S_V(\lambda)-E_N)P_N$ .

We write
\bea
P_N(S_V(\lambda)-E_N)P_N &=& P_N(S_V(\lambda)-E_N)P_N\Pi_N  \nonumber \\ && +   \;
 P_N(S_V(\lambda)-E_N)P_N\left(P_N-\Pi_N\right)
\eea
which in turn implies
\bea
&& P_N(S_V(\lambda)-E_N)P_N \left[ I - \left(P_N-\Pi_N\right)  \right]    \nonumber \\ &&
\phantom{xxxxxxxxxxxzxxxxxx} =
P_N(S_V(\lambda)-E_N)\Pi_N = P_NW(\lambda)\Pi_N.
\eea
For $q>33/2$, we have from Theorem \ref{thm:stability1} that
$ I - \left(P_N-\Pi_N\right) $ is invertible for $N$ sufficiently large.
Moreover $\|\left[ I - \left(P_N-\Pi_N\right)\right]^{-1}-I\|=O(N^{-\sigma})$ with
$\|P_N-\Pi_N\|=O(N^{-\sigma})$ and $\sigma=\frac{2q-33}{5}>0$.  Thus we have
\bea
P_N(S_V(\lambda)-E_N)P_N &=& \Pi_NW(\lambda)\Pi_N + \left(P_N-\Pi_N\right)W(\lambda)\Pi_N
\nonumber \\
&&+
P_NW(\lambda)\Pi_N\left\{  \left[ I - \left(P_N-\Pi_N\right)\right]^{-1}-I    \right\},  \nonumber \\
&=& \Pi_NW(\lambda)\Pi_N  + O(N^{-\sigma})W(\lambda)\Pi_N  \nonumber \\ &&+ P_NW(\lambda)\Pi_NO(N^{-\sigma}) . \label{est-shift}
\eea
 From Lemma \ref{norm-estimates-N}, we have
 $\|L_3\Pi_N\|=\|\Pi_N{L}_z\Pi_N\|=O(N)$ and $\|\left(x^2+y^2\right)\Pi_N\|=
 \|\left(\left(x^2+y^2\right)\Pi_N\right)^*\left(x^2+y^2\right)\Pi_N\|^{1/2}=\|\Pi_N\left(x^2+y^2\right)^2\Pi_N\|^{1/2}=O(N^4)$
 which implies $\|  W(\lambda)\Pi_N\| = O\left(h^2\epsilon(h)\right)$.  Hence, we have from equation (\ref{est-shift})
 \bea
&&\frac{P_N(S_V(\lambda)-E_N)P_N}{h^2\epsilon(h)} =  \frac{\Pi_NW(\lambda)\Pi_N}{h^2\epsilon(h)} +  O(N^{-\sigma})
= \Pi_N\left(-\frac{B}{2}hL_3\right)\Pi_N \phantom{xxxx}\nonumber \\
&& \phantom{xxxxxxxxxx}+ O(\epsilon(h)) +
 O(N^{-\sigma})  = \Pi_N\left(-\frac{B}{2}hL_3\right)\Pi_N +  O(N^{-\sigma}).
 \phantom{xxx} \label{est-normalized}
 \eea
Since $\|\Pi_N\left(-\frac{B}{2}hL_3\right)\Pi_N\|=O(1)$ then
$\left\|\frac{P_N(S_V(\lambda)-E_N)P_N}{h^2\epsilon(h)}\right\|=O(1)$.  Due to the self-adjointness of  $\frac{P_N(S_V(\lambda)-E_N)P_N}{h^2\epsilon(h)}$,
then the spectral radius of $\frac{P_N(S_V(\lambda)-E_N)P_N}{h^2\epsilon(h)}$ is
$O(1)$ as well.
Thus we see that the size of the  eigenvalue shifts $\tilde\nu_{N,j}$ is $O(h^2\epsilon(h))$.  Moreover, equation (\ref{est-normalized}) is the
basis to establish the following theorem:

\begin{theorem}\label{trace-aprox-theorem}
 Let $h=1/(N+1)$ and $\sigma = (2q-33)/5$. For any polynomial $Q$, we have
 \bea\label{eq:average-trace}
  \frac{1}{d_N}\sum_{j=1}^{d_N} Q \left(\frac{\tilde\nu_{N,j}}{h^2\epsilon(h)}\right) =
 \frac{1}{d_N} {\rm Tr}\; Q \left(\Pi_N\left(-\frac{B}{2}hL_3\right)\Pi_N\right) + O(N^{-\sigma}).
 \eea
 So for $q > 33/2$, the remainder term in \eqref{eq:average-trace} vanishes as $N \rightarrow \infty$.
\end{theorem}

\begin{proof}
For $N$ sufficiently large,  there exist a fixed  interval $\left[-A,A\right]$ containing  both
 $\frac{\tilde\nu_{N,j}}{h^2\epsilon(h)}$, $ j=1,\ldots,d_N$, and the eigenvalues of  the operator  $\Pi_N\left(-\frac{B}{2}hL_3\right)\Pi_N$.
Since a polynomial is a finite linear combination of monomials, then we only need to show equation  (\ref{eq:average-trace}) for a monomial.
We consider a monomial of degree $k \in \N$ and write
\beq\label{ident-trace}
\frac{1}{d_N}\sum_{j=1}^{d_N}\left(\frac{\tilde\nu_{N,j}}{h^2\epsilon(h)}\right)^k =
 \frac{1}{d_N} {\rm Tr}\;\left(\frac{P_N\left(S_V(\lambda)-E_N\right)P_N}{h^2\epsilon(h)}\right)^k.
\eeq
To simplify notation, let $A_N :=\Pi_N\left(-\frac{B}{2}hL_3\right)\Pi_N$. From equation (\ref{est-normalized}), we have
\bea\label{eq:substitute1}
&& \left(\frac{P_N(S_V(\lambda)-E_N)P_N}{h^2\epsilon(h)}\right)^k =  P_N
\left(A_N\right)^k +  P_NO(N^{-\sigma}), \phantom{xxx}\nonumber \\
&& \phantom{xxxxx} = \left(A_N\right)^k +
\left(P_N-\Pi_N\right)\left(A_N\right)^k  +
P_NO(N^{-\sigma}).
\eea
To evaluate the trace of the second term on the right of the last line of \eqref{eq:substitute1}, consider an orthonormal  basis $\{\phi_j\}_{j=1}^{d_N}$ for the range of $\Pi_N$ and extend it to an orthonormal  basis  $\{\phi_j\}_{j=1}^{\infty}$
of $L^2(\R^3)$. Then,  we have
\bea
&&\left|\frac{1}{d_N} {\rm Tr}\;
\left(P_N-\Pi_N\right)A_N^k\right| =
\left|\frac{1}{d_N}\sum_{j=1}^\infty \langle \phi_j,\left(P_N-\Pi_N\right)A_N^k\phi_j \rangle \right|,
\phantom{xxxxxxxxxx} \nonumber \\
&& \phantom{xxxx}= \left|\frac{1}{d_N}\sum_{j=1}^{d_N} \langle \phi_j,\left(P_N-\Pi_N\right)A_N^k\phi_j \rangle \right| \leq \|P_N-\Pi_N\| = O(N^{-\sigma}).
\eea
To evaluate the trace of the third term on the right of the last line of \eqref{eq:substitute1}, consider an orthonormal basis $\{\Psi_j\}_{j=1}^{d_N}$ for the range of $P_N$
and extend it to an orthonormal basis $\{\Psi_j\}_{j=1}^{\infty}$  of $L^2(\R^3)$. Thus, we have
\bea
&&\left|\frac{1}{d_N} {\rm Tr}\;\left(P_NO(N^{-\sigma}) \right)\right| =
\left|\frac{1}{d_N}\sum_{j=1}^\infty \langle O(N^{-\sigma}) \Psi_j, P_N\Psi_j \rangle\right|,
\phantom{xxxxxxxxxxxxx}\nonumber \\
&&
\phantom{xxxxxxxxxx} = \left|\frac{1}{d_N}\sum_{j=1}^{d_N} \langle O(N^{-\sigma}) \Psi_j, P_N \Psi_j \rangle \right| \leq  O(N^{-\sigma}) . \label{est-norm-O}
\eea
From equations (\ref{ident-trace})-(\ref{est-norm-O}), the theorem follows for a monomial and hence for any polynomial.
\end{proof}



\section{A trace estimate for the Zeeman perturbation of the hydrogen atom.}\label{sec:poly1}

In this section, we calculate the large $N$ limit of the
trace on the right side of \eqref{eq:average-trace}
involving the third component $L_3$ of the angular momentum. This extends previous results to a family of perturbations that are first-order differential operators. In section \ref{subsec:ang-mom-term}, we prove Proposition \ref{estinnerprod} that presents a large
$N$ estimate  for the expected value of powers of the operator $h \frac{B}{2} L_3$, $h=1/(N+1)$, on a coherent state $\Psi_{\alpha,N}$. This key technical result is used in section \ref{subsec:szego-type-ang-mom1} to prove the Szeg\"o-type theorem for the angular momentum operator, Theorem \ref{szego-theorem}. Finally, the proofs of the main theorems are given in section \ref{subsec:proof-main1-2}.

\subsection{The angular momentum perturbation term.}\label{subsec:ang-mom-term}

In this section, we prove a Szeg\"o-type theorem for the operator
$h \frac{B}{2} L_3$, where $L_3$ is the angular momentum operator in the
third direction, namely,
$L_3=-\imath\left(x\frac{\partial\phantom{x}}{\partial{y}}-
y\frac{\partial\phantom{x}}{\partial{x}}\right)$.  The  Szeg\"o-type theorem, Theorem \ref{szego-theorem}, describes the limit $N \rightarrow \infty$ of the normalized  trace indicated in the right-hand side of equation (\ref{eq:average-trace}) in Theorem \ref{trace-aprox-theorem}.
As in section \ref{sec:trace1}, it is enough to consider the case $\rho(x)=x^m$, $m\in\N^*$,  on a compact interval (see proof of Theorem  \ref{szego-theorem} below).

The key idea in the proof of Theorem  \ref{szego-theorem} is to use, for each $N$, the  system of coherent states
$\{\Psi_{\alpha,N}\}_{\alpha\in{\mathcal A}}$ (see section \ref{sec:app2-coherentstates1}, Appendix 2).
Such a system
gives a resolution of the identity on the range of $\Pi_N$ (see equation (\ref{eq:proj-coherentstates1}))
and,  as a consequence, one can express the trace in equation
(\ref{eq:average-trace}),  with $\rho(x)=x^m$,  as an integral
involving the inner products $\left<\Psi_{\alpha, N},\left((-h{B}/2) L_3\right)^m\Psi_{\alpha, N}\right>$ with respect to the measure $d\mu(\alpha)$ (see equation (\ref{traza})).

We estimate that inner product in Proposition \ref{estinnerprod}. This result goes beyond our work with the Stark hydrogen atom Hamiltonian \cite{hVB1} as we must estimate the derivatives of the coherent states.
We first state the proposition but defer its proof until after the key technical Lemma \ref{CS-der}.



\begin{proposition}\label{estinnerprod}
For $\alpha\in{\mathcal A}$, an integer  $m > 0$,
and $h=1/(N+1)$, we have for $N\rightarrow{\infty}$,
\beq\label{eq:trace-reduction1}
\left<\Psi_{\alpha, N},\left(h\tilde{B}L_3\right)^m\Psi_{\alpha, N}\right>
= \left(\tilde{B} \ell_3(\alpha)\right)^m + O(N^{-1}),
\eeq
where
$\ell_3(\alpha) = \Re(\alpha)_1\Im(\alpha)_2-\Re(\alpha)_2\Im(\alpha)_1$ and $\tilde{B}=-B/2$.
\end{proposition}

\begin{remark}
When $\Re \alpha \in{\Sp}^3_o$ (the 3-sphere with the north pole removed),
we have that  $\ell_3(\alpha)$
is the angular momentum in the third direction associated to the point
$({\bf x,p})={\mathcal M}^{-1}\circ\sigma(\alpha)$, where ${\mathcal M}^{-1}$ is the
inverse of the Moser map ${\mathcal
M}:T^*(\mathbb{R}^n)\rightarrow{T}^*(\Sp^3_o)$ and $\sigma$ is the map defined in Eq.
(\ref{eq:sigma1})  (see Appendix 1, section \ref{sec:moser-map}).
\end{remark}

Before presenting the proof, we present the lemma on the decay properties of the derivatives of the coherent states as $N \rightarrow \infty$.
We recall some notation from Appendix 2, section \ref{sec:app2-coherentstates1}. For $\alpha \in \mathcal{A}$, a coherent state is defined by
\beq\label{eq:coh-st1}
\Psi_{\alpha, N}={\mathcal F}^{-1}D_{r_{N}}J^{1/2}K\Phi_{\alpha, N},
\eeq
with $\Phi_{\alpha,N}(\omega) = a(N) ( \alpha \cdot \omega)^N$, where $\mathcal{F}$ is the Fourier transform (see Eq.  (\ref{defft})),
the operator $K$ is the unitary operator defined in Eq. (\ref{defK} )and the operator $J$
is the self-adjoint defined in Eq. (\ref{defJ}), the operator $D_\sigma$ is the dilation operator defined in \eqref{eq:dilation0}, and $\omega \in {\Sp}^3$. We also define
$\tilde{L}_3= \mathcal{F}L_3 \mathcal{F}^{-1}  = -\imath\left[p_2\frac{\partial\phantom{p_1}}{\partial{p_1}}-
p_1\frac{\partial\phantom{p_2}}{\partial{p_2}}\right].$

\begin{lemma}\label{CS-der}
Let $m\geq{1}$ be an integer and $\alpha\in{\mathcal A}$.   Then for $h=1/(N+1)$ we have
\bea\label{Lzm}
&&\left(h\tilde{B}\tilde{L}_3\right)^mJ^{1/2}K\Phi_{\alpha, N}(\vecp) \phantom{xxxxxxxxxxxxxx}\nonumber \\
&&\phantom{xxx}
=
\frac{\left(-\imath\tilde{B}\right)^m\left\{\;\left[\alpha_1p_2
-\alpha_2p_1\right]^{m}+Q_{\alpha,m}\left(\vecp;N\right)\;\right\}}{\left(\frac{|\vecp|^2+1}{2}\right)^{m+2}}
 \phantom{x} \left[ \frac{\Phi_{\alpha, N}(\omega(\vecp))}{\left(\alpha\cdot\omega(\vecp)\right)^{m}} \right],
 \nonumber \\
 &&
\eea
where the function
\beq
Q_{\alpha,m}(\vecp;N)=\sum_{\ell=0}^{m-1}P_{\alpha,m,\ell}(p_1,p_2;N)\left(\frac{|\vecp|^2+1}{2}
\right)^\ell\left(\alpha\cdot\omega(\vecp)\right)^\ell \label{def-Q-type}
\eeq
 with
 \beq
 P_{\alpha,m,\ell}(p_1,p_2;N)= \sum_{j=0}^{m-\ell} C_{N,m,\ell,j}(N)\left[\alpha_1p_2-\alpha_2p_1\right]^{j}
 \left[\alpha_1p_1+\alpha_2p_2\right]^{m-\ell-j} \label{exppol}
 \eeq
and all the coefficients  $C_{N,m,\ell,j}(N)$ are   $O(N^{-1})$.
\end{lemma}

\noindent
{\bf Notation}: Given $m,\ell\in\mathbb{Z}$, $1\leq{m}$ and $0\leq\ell\leq{m-1}$, we will denote by $Pol_{\alpha,m,\ell}(p_1,p_2;N)$
 any  polynomial in $(p_1,p_2)$ that is a linear combination of
$\left[\alpha_1p_2-\alpha_2p_1\right]^{j}
 \left[\alpha_1p_1+\alpha_2p_2\right]^{m-\ell-j}$ with  $0\leq{j}\leq{m-\ell}$
 and such all the coefficients depend only on $N$ and are $O(N^{-1})$.  A polynomial
 $Pol_{\alpha,m,\ell}(p_1,p_2;N)$
 will be called a
 $(m,\ell)$-type polynomial.   Thus  $P_{\alpha,m,\ell}(p_1,p_2;N)$ in equation (\ref{exppol})
  is a $(m,\ell)$-type polynomial.

\begin{proof}
1. The following equalities hold for all integers $k,q\geq{1}$:
\bea
&&\tilde{L}_3\left(\frac{2}{|\vecp|^2+1}\right)^k = 0
,     \nonumber \\
&&\tilde{L}_3\left(\alpha\cdot\omega(\vecp)\right)^q = (-\imath)q\left(\alpha\cdot\omega(\vecp)\right)^{q-1}\left(\frac{2}{|\vecp|^2+1}\right)\left(\alpha_1p_2 - \alpha_2p_1\right),    \nonumber \\
&& \tilde{L}_3\left(\alpha_1p_2-\alpha_2p_1\right) = \imath\left(\alpha_1p_1+\alpha_2p_2\right),  \nonumber \\
&& \tilde{L}_3\left(\alpha_1p_1+\alpha_2p_2\right)  = (-\imath)\left(\alpha_1p_2-\alpha_2p_1\right).
\eea
We also note that $\tilde{L}_3$ is dilation invariant.

\noindent
2. We prove Lemma \ref{CS-der} by induction.
We write (see Appendix 2, section \ref{sec:app2-coherentstates1}),
\beq
J^{1/2}K\Phi_{\alpha, N}(\vecp) = a(N)\left(\frac{2}{|\vecp|^2+1}\right)^2\left(\alpha\cdot\omega(\vecp)\right)^N.
\eeq
First note that for $m=1$
\bea
&&\left(h\tilde{B}\tilde{L}_3\right)J^{1/2}K\Phi_{\alpha, N}(\vecp) = \nonumber \\
&&  \phantom{xxxxx}\frac{\left(-\imath\tilde{B}\right)\left(\;\left[\alpha_1p_2
-\alpha_2p_1\right]+ \frac{1}{N}\left[\alpha_1p_2
-\alpha_2p_1\right]\;\right)}{\left(\frac{|\vecp|^2+1}{2}\right)^{3}}
 \phantom{x}\frac{\Phi_{\alpha, N}(\omega(\vecp))}{\left(\alpha\cdot\omega(\vecp)\right)}.
\eea
Since $\frac{1}{N}\left[\alpha_1p_2
-\alpha_2p_1\right]$
is a $(m=1,\ell=0)$-type polynomial, then it can be written in the form
indicated in the right hand side of equation (\ref{def-Q-type}).  Thus Lemma \ref{CS-der} is valid for $m=1$.

\noindent
3. Let us assume that equation (\ref{Lzm}) holds for an integer $m > 1$. Since  {\bf (i)} $\tilde{L}_3P_{\alpha,m,\ell}(p_1,p_2;N)$ is a $(m,\ell)$-type polynomial and then
a $(m+1,\ell+1)$-type polynomial as well (denoted by $Pol_{\alpha,m+1,\ell+1}(p_1,p_2;N)$)
and we also have {\bf (ii)} $P_{\alpha,m,\ell}(p_1,p_2;N)\left(\alpha_1p_2
-\alpha_2p_1\right)$ is a $(m+1,\ell)$-type  polynomial (denoted by $Pol_{\alpha,m+1,\ell}(p_1,p_2;N)$) then
\bea
\left(h\tilde{B}\tilde{L}_3\right)^{m+1}J^{1/2}K\Phi_{\alpha, N}(\vecp) =
\frac{\left(-\imath\tilde{B}\right)^{m+1}}{\left(\frac{|\vecp|^2+1}{2}\right)^{m+3}}
\frac{\Phi_{\alpha, N}(\omega(\vecp))}{\left(\alpha\cdot\omega(\vecp)\right)^{m+1}}
\phantom{xxxxxxxxxxxxx}
\nonumber \\
\cdot\left\{\left[\alpha_1p_2
-\alpha_2p_1\right]^{m+1}
 -\left(\frac{m+1}{N+1}\right)\left[\alpha_1p_2
-\alpha_2p_1\right]^{m+1} \right.
\phantom{xxxxxxxxxxxxxxx}
\nonumber \\
- \left.  \left(\frac{m}{N+1}\right)\left[\alpha_1p_2
-\alpha_2p_1\right]^{m-1}\left[\alpha_1p_1+\alpha_2p_2\right]\left(\frac{|\vecp|^2+1}{2}
\right)\left(\alpha\cdot\omega(\vecp)\right)\right.
\phantom{x}
\nonumber \\
+ \left.  \frac{1}{N+1}\sum_{\ell=0}^{m-1}Pol_{\alpha,m+1,\ell+1}(p_1,p_2;N)\left(\frac{|\vecp|^2+1}{2}
\right)^{\ell+1}\left(\alpha\cdot\omega(\vecp)\right)^{\ell+1}\right.
\phantom{xx}
\nonumber \\
+\left. \frac{1}{N+1}\sum_{\ell=0}^{m-1}\left(N-m+\ell\right)Pol_{\alpha,m+1,\ell}(p_1,p_2;N)\left(\frac{|\vecp|^2+1}{2}
\right)^{\ell}\left(\alpha\cdot\omega(\vecp)\right)^{\ell}\right\}.
\eea
One can check that each one of the last four terms inside the bracket can be written
as indicated in the right hand side of equation (\ref{def-Q-type}) and then altogether give us the function
$Q_{\alpha,m+1}(\vecp;N)$ in equation (\ref{def-Q-type}).
\end{proof}

\noindent
We finish this section with the proof of Proposition \ref{estinnerprod}.

\begin{proof}
1. Proposition \ref{estinnerprod} is obviously valid for $m=0$, so we assume $m\geq{1}$.
Using the definition of $\Psi_{\alpha, N}$ above Lemma \ref{CS-der} in \eqref{eq:coh-st1},
the inner product is
\bea
&&\left<\Psi_{\alpha,N},\left(h\tilde{B}L_3\right)^m\Psi_{\alpha, N}\right>
_{L^2(\mathbb{R}^3)}\nonumber \\
&&\phantom{xxxxxxxx}= \left<\Phi_{\alpha, N},
K^{-1}J^{1/2}\left(h\tilde{B}\tilde{L}_3\right)^mJ^{1/2}K\Phi_{\alpha,
N}\right>_{L^2(\Sp^3)},
\eea
where $\tilde{L}_3= \mathcal{F}L_3 \mathcal{F}^{-1}  = -\imath\left[p_2\frac{\partial\phantom{p_1}}{\partial{p_1}}-
p_1\frac{\partial\phantom{p_2}}{\partial{p_2}}\right]$ and $D^{-1}_{r_{N}}\tilde{L}_3D_{r_{N}}=\tilde{L}_3$.

\noindent
2. Using Lemma \ref{CS-der}, we have
\begin{eqnarray}
K^{-1}J^{1/2}\left(h\tilde{B}\tilde{L}_3\right)^mJ^{1/2}K\Phi_{\alpha, N} (\omega)  = \phantom{xxxxxxxxxxxxxxxxxxxxxxxxxxxx} \nonumber \\
\frac{\left(-\imath\tilde{B}\right)^m
\left\{\;\left[\alpha_1p_2(\omega)
-\alpha_2p_1(\omega)\right]^{m}+Q_{\alpha,m}(\vecp(\omega);N)\;\right\}}
{\left(\frac{|\vecp|^2(\omega)+1}{2}\right)^{m+1}}
 \left( \frac{\Phi_{\alpha, N}(\omega)}{\left(\alpha\cdot\omega\right)^{m}}\right),  \phantom{x}
\end{eqnarray}
with  $Q_{\alpha,m}(p;N)$ is defined in \eqref{def-Q-type} and $P_{\alpha,m,\ell}(p_1,p_2;N)$ is defined in \eqref{exppol}.

\noindent
3. Since $|\alpha|=\sqrt{2}$ then , for N sufficiently large,  there exist a constant C such that
$|P_{\alpha,m,\ell}(p_1,p_2;N)|\leq{C}|p|^{m-\ell}/N$ uniformly in $\alpha$.
Since $|\left(\alpha\cdot\omega(\vecp)\right)|\leq{1}$ for all $\alpha\in{\mathcal A}$ and $p\in\mathbb{R}^3$ then we have $|Q_{\alpha,m}(p;N)|\leq{C}|p|^{2m-1}/N$ uniformly in
$\alpha\in{\mathcal A}$
  which in turn implies
\beq
\frac{Q_{\alpha,m}(\vecp(\omega);N)}{\left(\frac{|\vecp|^2(\omega)+1}{2}\right)^{m+1}}=O(N^{-1})
\phantom{xxx}
\rm{uniformly} \phantom{x}\rm{in} \phantom{x}\omega\in{\Sp^3} \phantom{x}\rm{and} \phantom{x}\alpha\in{\mathcal A}.  \label{esterror}
\eeq
Now notice that for $m$ fixed and $d\Omega$ denoting the
surface measure on the 3-sphere:
\bea
\int\frac{|\Phi_{\alpha,
N}(\omega)|^2}{\left|\alpha\cdot\omega\right|^{m}}d\Omega(\omega) \leq
a^2(N)\int{\left|\alpha\cdot\omega\right|^{2(N-m)}}d\Omega(\omega)=
\frac{a^2(N)}{a^2(N-m)}.
\eea

\noindent
4. From Proposition \ref{estnormacs} in Appendix 2,
we  have  $\int{|\Phi_{\alpha,
N}(\omega)|^2}/{|\left(\alpha\cdot\omega\right)^{m}|}d\Omega(\omega)=O(1)$ for $m$ fixed and $N\rightarrow\infty$. Thus using the expression for
the inverse of the stereographic projection $p=\vecp(\omega)$  (see equation (\ref{invmoser}))  and equation (\ref{esterror}) we obtain
\bea \label{prodint}
&&\left<\Psi_{\alpha,N},\left(h\tilde{B}L_3\right)^m\Psi_{\alpha, N}\right>
_{L^2(\mathbb{R}^3)} = \left(-\imath\tilde{B}\right)^m{a}^2(N)
 \nonumber \\
&& \phantom{xxx}\cdot\int\exp\left(\imath{N}\phi_{\alpha}(\omega)\right)
\frac{\left[\alpha_1p_2(\omega)
-\alpha_2p_1(\omega)\right]^{m}}{\left(\frac{(|\vecp|^2(\omega)+1)^2}{2}\right)^{m+1}
\left(\alpha\cdot\omega\right)^{m}} d\Omega(\omega) + O(N^{-1}),
\nonumber \\
&&\phantom{xx}=\left(-\imath\tilde{B}\right)^m{a}^2(N)\nonumber \\
&&\phantom{xxx}\cdot
\int\exp\left(\imath{N}\phi_{\alpha}(\omega)\right)
\left(1-\omega_4\right) \frac{\left[\alpha_1\omega_2
-\alpha_2\omega_1\right]^{m}}{\left(\alpha\cdot\omega\right)^{m}}
d\Omega(\omega) + O(N^{-1}),
\eea
where the phase function is given by $ \phi_{\alpha}(\omega)=
-\imath\ln\left(|\alpha\cdot\omega|^2\right).$

\noindent
5. Let us denote by $\vecai=\Re \alpha$ and $\vecbi=\Im{\alpha}$. In
order to evaluate the integral in (\ref{prodint}), we
introduce the following coordinates for $\Sp^3$.   Consider two
orthonormal vectors $\vece_3$ and $\vece_4$ such that
$\{\vecai,\vecbi,\vece_3,\vece_4\}$ is an orthonormal basis of
$\mathbb{R}^4$. Except for a set of measure zero, any element
$\omega\in\Sp^3$ can be written as
\beq
\omega=\sqrt{1-z_3^2-z_4^2}\cos(\theta)\vecai +
\sqrt{1-z_3^2-z_4^2}\sin(\theta)\vecbi+z_3\vece_3+z_4\vece_4
\eeq
with $\theta\in[0,2\pi]$ and $z=(z_3,z_4)$ in the unit disk
$z_3^2+z_4^2<1$.

\noindent
6. The surface measure is given  by   $d\Omega(\omega)=
d\theta{d}z_3dz_4$.  Now we write the integral in
(\ref{prodint}) as an iterated integral estimating the integration
with respect to the variable $z$  by  using the stationary phase
method. Note that
$\phi_{\alpha}(\omega)=-\imath\ln(1-z_1^2-z_2^2)$ is independent
of $\theta$.  Since $z=0$ is a non-degenerate critical point of
the function $\phi_{\alpha}$ and
$\left[\det\left(\frac{N\phi_{\alpha}^{\prime\prime}(0)}{2\pi\imath}\right)
\right]^{-1/2}=\frac{\pi}{N}$, with $\phi_{\alpha}^{\prime\prime}$
denoting the Hessian matrix of the function $\phi_{\alpha}$ with
respect to the variable $z$, then
\begin{eqnarray} &&\left<\Psi_{\alpha,
N},\left(h\tilde{B}L_3\right)^m\Psi_{\alpha, N}\right>
_{L^2(\mathbb{R}^3)} = \frac{\left(-\imath\tilde{B}\right)^m}{2\pi}
\int_0^{2\pi}\frac{\left[1-\vecai_4\cos(\theta)-\vecbi_4\sin(\theta)\right]}{\exp(\imath{m}\theta)}
 \nonumber
\\ && \phantom{xx} \times \left[\alpha_1\{\vecai_2\cos(\theta)+\vecbi_2\sin(\theta)\}-
\alpha_2\{\vecai_1\cos(\theta)+\vecbi_1\sin(\theta)\}\right]^md\theta
+ O(N^{-1}) \nonumber \\
&& \phantom{xx} =
\tilde{B}^m\left[\vecbi_1\vecai_2-\vecai_1\vecbi_2\right]^m + O(N^{-1}), \label{calculoint}
\end{eqnarray}
where we have used the estimate
$a^2(N)=\frac{N}{\pi}\left[\frac{1}{2\pi}+O(N^{-1})\right]$  (see
Appendix 2, section \ref{sec:app2-coherentstates1}). The last equality in equation (\ref{calculoint}) can be obtained by noting that the coefficients of $\exp(\imath{m}\theta)$ and $\exp(\imath(m-1)\theta)$ in the expansion of
the factor $\left[\alpha_1\{\vecai_2\cos(\theta)+\vecbi_2\sin(\theta)\}-
\alpha_2\{\vecai_1\cos(\theta)+\vecbi_1\sin(\theta)\}\right]^m$ are $\left(\imath\left[\vecbi_1\vecai_2-\vecai_1\vecbi_2\right]\right)^m$ and zero respectively.
This leads to the expression on the right side of \eqref{eq:trace-reduction1}.
\end{proof}


\subsection{Szeg\"o-type theorem for the angular momentum operator.}\label{subsec:szego-type-ang-mom1}

Using the resolution of the identity given by the coherent states $\Psi_{\alpha,
N}$ (see equation (\ref{eq:proj-coherentstates1})), Proposition \ref{estinnerprod},  and the commutativity of $L_3$ with the scaled hydrogen atom Hamiltonian $S_V$ we have the following Szeg\"o -type theorem:

\begin{theorem}[Szeg\"o theorem]\label{szego-theorem}
Let $\rho: \R \rightarrow \R$ be continuous. Then, for $h=1/{(N+1)}$ we have
\bea \label{eq:szego-theorem1}
\lim_{N \rightarrow \infty} \frac{1}{d_N} Tr\left(
\rho\left(\Pi_N\left( -\frac{B}{2}h L_3\right)\Pi_N\right)\right) =
 \int_{{\mathcal A}} \rho \left( -\frac{B}{2} \ell_3 (\alpha)
  \right) ~ d\mu(\alpha) \label{szego1} \\
=\int_{\Sigma(-1/2) } \rho \left(  -\frac{B}{2} \ell_3 ({\bf x, p})
  \right) ~ d\mu_L ({\bf x, p}) \label{eq:szego22}
  \eea
where $\ell_3 ({\bf x,p}) = x p_y - y p_x$ is the z-component of the classical angular
momentum assigned to the point $({\bf x,p})\in\Sigma(-1/2)$.
\end{theorem}

\begin{proof}
1. As in the proof of Theorem \ref{trace-aprox-theorem}, we first prove \eqref{eq:szego22} for a polynomial $Q$.
Consequently, it is enough to show Theorem
\ref{szego-theorem} for a monomial $Q(x)=x^m$.
Since $L_3$ commutes with the Hamiltonian $H_V$, then $L_3$ commutes
with the projections $\Pi_N$ which in turn implies
$\left(\Pi_N
\left(-\frac{B}{2}h\right)L_3\Pi_N\right)^m=
\Pi_N\left(-\frac{B}{2}hL_3\right)^m\Pi_N$.
Thus we have from Proposition \ref{estinnerprod}
\bea\label{evaltrace}
&&\frac{1}{d_N} Tr\left(\left(\Pi_N\left(-\frac{B}{2}h\right)
L_3\Pi_N\right)^m\right) = \frac{1}{d_N} Tr\left(\Pi_N\left(-\frac{B}{2}h
L_3\right)^m\Pi_N\right) \nonumber \\
&& \phantom{xxxxxxxxxxxxxxxx} = \int_{\alpha\in{\mathcal A}} \left<\Psi_{\alpha, N},\left(-\frac{B}{2}hL_3\right)^m\Psi_{\alpha,
N}\right> _{L^2(\mathbb{R}^3)}  d\mu(\alpha) \nonumber \\
&&  \phantom{xxxxxxxxxxxxxxxx} = \int_{\alpha\in{\mathcal A}}
\left(-\frac{B}{2}\ell_3(\alpha)\right)^m d\mu(\alpha)+ O(N^{-1}) ,
 \eea
where we have used that the measure $\int_{\mathcal{A}} ~d\mu = 1$ and that for any linear operator $T:{\mathcal E}_N\rightarrow{\mathcal E}_N$, with ${\mathcal E}_N$ denoting the range of $\Pi_N$, we have
\bea
Tr(T) = d_N\int_{\alpha\in{\mathcal A}} \left<\Psi_{\alpha, N},T\Psi_{\alpha, N}\right>   \label{traza}
d\mu(\alpha) .
\eea
Taking the limit $N\rightarrow\infty$ we get the first equality in Theorem
\ref{szego-theorem} for monomials and hence for polynomials $Q$.

\noindent
2. To extend the result to continuous functions $\rho$, we note that
the operators $A_N =: \Pi_N \left( -\frac{B}{2} h B L_3 \right)\Pi_N$
satisfy $\| A_N \| \leq B/2$, so they are uniformly bounded in $N$. Similarly, for $({\bf x,p})\in\Sigma(-1/2)$, we have
$|\ell_3 ({\bf x, p})|\leq{1}$, so that $\left| - \frac{B}{2} \ell_3 (\alpha) \right| \leq \frac{B}{2}$. Consequently, we only need to consider continuous functions on the interval $[ - B/2, B/2]$. Given $\epsilon > 0$, there exists a polynomial $Q_\epsilon$ so that
\beq\label{eq:poly-approx1}
| \rho (x) - Q_\epsilon (x) | \leq \frac{ 2 \epsilon}{B}, ~~~\forall x \in \left[ - \frac{B}{2}, \frac{B}{2} \right].
\eeq
This implies that for all $N$:
\beq \label{eq:trace-approx1}
\| \rho(A_N)-Q_\epsilon(A_N)\| \leq \epsilon,
\eeq
which in turn implies that for all $N$:
\beq \label{eq:trace-approx2}
\left|\frac{1}{d_N} Tr\left(
\rho\left(A_N\right)\right) - \frac{1}{d_N} Tr\left(
Q\epsilon\left(A_N\right)\right)\right| \leq \epsilon .
\eeq
It also follows from \eqref{eq:poly-approx1} that
\beq
 \left| \int_{{\mathcal A}} \rho \left( -\frac{B}{2} \ell_3 (\alpha)
  \right) ~ d\mu(\alpha) -\int_{{\mathcal A}} Q_\epsilon \left( -\frac{B}{2} \ell_3 (\alpha)
  \right) ~ d\mu(\alpha) \right| \leq \epsilon .
\eeq
since $\int_{\mathcal A} ~ d\mu(\alpha) = 1$. By \eqref{eq:trace-approx1}--\eqref{eq:trace-approx2}, we have
\bea\label{eq:trace-approx3}
\left| \frac{1}{d_N} Tr\left(
\rho\left(A_N\right)\right) -  \int_{{\mathcal A}} \rho \left( -\frac{B}{2} \ell_3 (\alpha)
  \right) ~ d\mu(\alpha)\right| \phantom{xxxxxxxxxxxxxxxxx}\nonumber \\
  \leq 2\epsilon + \left| \frac{1}{d_N} Tr\left(
Q_\epsilon\left(A_N\right)\right) - \int_{{\mathcal A}} Q_\epsilon \left( -\frac{B}{2} \ell_3 (\alpha)
  \right) ~ d\mu(\alpha)\right|.
\eea
By Theorem \ref{trace-aprox-theorem}, the limit as $N \rightarrow \infty$ of the second term on the right side of \eqref{eq:trace-approx3} is zero. Since $\epsilon > 0$ is arbitrary, it follows that the limit as $N \rightarrow \infty$ of the first term on the left side of \eqref{eq:trace-approx3} is zero establishing the first equality of
\eqref{eq:szego22}.

\noindent
3. As for the second equality in \eqref{eq:szego22}, we recall
that ${\mathcal A}$ can be identified with the unit
cotangent bundle $T_1^*(\Sp^3)$ of the 3-sphere by the map
$\sigma(\alpha)=\left(\Re \alpha,-\Im \alpha\right)$.
The inverse of the Moser map ${\mathcal M}^{-1}$  (see Appendix 1)
identifies in turn $T_1^*(\Sp_o^3)$ with the surface of constant
energy $\Sigma(-1/2)$ for the Kepler problem. Hence we have
$\ell_3(\alpha)=\ell_3({\bf x,p})$ where $({\bf x,p})={\mathcal
M}^{-1}\circ\sigma(\alpha)$. Thus the second equality of Theorem
\ref{szego-theorem}  is consequence of equation \eqref{szego1} and
Proposition \ref{liouville} below.
\end{proof}

\begin{proposition}\label{liouville}
The Liouville measure
$d\mu_L$ on $\Sigma(-1/2)$ is the push-forward measure of $d\tilde\mu$
by the map ${\mathcal M}^{-1}\circ\sigma$ with $d\tilde\mu$ the
restriction of $d\mu$ to the set $\tilde{{\mathcal A}}=\{\alpha \in{\mathcal A} \; | \; \Re \alpha \neq(0,0,0,1)\}$ .
\end{proposition}

\begin{proof}
1. We have $d\mu$ defined as the normalized SO(4)-invariant measure on ${\mathcal A}$.
Note that the push-forward measure $\sigma_*(d\mu)$ is the normalized
SO(4)-invariant measure on $T_1^*\Sp^3$.
Let us consider the normalized Liouville measure $d\nu_L$
on $T_1^*\Sp^3$ obtained  from the symplectic form
$\sum_{k=1}^{4}d\xi_k\wedge{d}\omega_k \;| \; _{T^*\Sp^3}$
that endows  $T^*\Sp^3$  with a symplectic structure. It is given by the restriction  to $T^*\Sp^3$ of the
canonical symplectic form  $\sum_{k=1}^{4}d\xi_k\wedge{d}\omega_k$ for the ambient $T^*\R^4$ .  Here
$(\omega,\xi)$ are canonical coordinates for $T^*\R^4$ and the condition
$|\omega|=1$ specifies the 3-sphere $\Sp^3$.
Note that $d\nu_L$  is SO(4)-invariant.  Therefore the measure
$\sigma_*(d\mu)$ coincides with  $d\nu_L$.

\noindent
2. On the other hand, the Moser map ${\mathcal M}$ is a canonical transformation (symplectic diffeomorphism) from $T^* \R^3$ onto $T^* \Sp_o^3$ and ${\mathcal M}(\Sigma(-1/2))= T_1^* \Sp_o^3$. Thus the restriction to $T_1^* \Sp_o^3$ of the measure
$d\nu_L$ is actually the push-forward measure of $d\mu_L$ under the Moser map.  Since $d\nu_L=\sigma_*(d\mu)$ then the restriction $d\tilde{\mu}$ of $d\mu$ to the set $\tilde{{\mathcal A}}$
must coincide
with $(\sigma^{-1}\circ{\mathcal M})_*d\mu_L$.  From the equality
$\mu\left(\{\alpha \in{\mathcal A} \; | \; \Re \alpha =(0,0,0,1)\}\right)=0$, we have that for any  integrable function $f:\Sigma(-1/2)\rightarrow\R$:
\bea
\int_{\Sigma(-1/2)} f(x,p) d\mu_L(x,p) = \int_{\tilde{{\mathcal A}}} f\circ{\mathcal M}^{-1}\circ\sigma(\alpha)
d\tilde{\mu} = \int_{\mathcal A} \hat{f}(\alpha) d\mu   \label{liou-quadric}
\eea
where $ \hat{f}$ is an extension of $f\circ{\mathcal M}^{-1}\circ\sigma$  with domain $\tilde{{\mathcal A}}$ to  the set ${\mathcal A}$.
\end{proof}


\subsection{Conclusion of the proof of Theorems  \ref{thm:main1}    and
\ref{thm:main2}.}\label{subsec:proof-main1-2}

For any continuous function $\rho$ and $\epsilon > 0$, there is a polynomial $Q_\epsilon$ that uniformly approximates $\rho$ on $\left[ - \frac{B}{2}, \frac{B}{2} \right]$ as above.
It follows from Theorem \ref{trace-aprox-theorem} that for any continuous $\rho$, we have
\bea
\lefteqn{\left| \frac{1}{d_N}\sum_{j=1}^{d_N}\rho\left(\frac{\tilde\nu_{N,j}}{h^2\epsilon(h)}\right) - \frac{1}{d_N} {\rm Tr}
\rho (A_N) \right|}  \nonumber \\
 &\leq&  
 \frac{1}{d_N}\sum_{j=1}^{d_N}\left|\rho\left(\frac{\tilde\nu_{N,j}}{h^2\epsilon(h)}\right) - Q_\epsilon \left(\frac{\tilde\nu_{N,j}}{h^2\epsilon(h)}\right) \right| + \frac{1}{d_N} \left| {\rm Tr}
\rho  (A_N)  - {\rm Tr}
Q_\epsilon  (A_N) \right|\nonumber \\
  && + \frac{1}{d_N}\left| \sum_{j=1}^{d_N}Q_\epsilon \left(\frac{\tilde\nu_{N,j}}
{h^2\epsilon(h)}\right)- {\rm Tr}
Q_\epsilon  (A_N) \right|
 \leq 2\epsilon +O_{Q_\epsilon} (N^{- \sigma})  \label{desglose}
 \eea
 where we are using the notation $O_{Q_\epsilon} (N^{- \sigma})$ to emphasize that the  third term
 $ \frac{1}{d_N}\left| \sum_{j=1}^{d_N}Q_\epsilon \left(\frac{\tilde\nu_{N,j}}
{h^2\epsilon(h)}\right)- {\rm Tr}
Q_\epsilon  (A_N) \right| $ is $O(N^{- \sigma})$ but depends on the polynomial $Q_\epsilon$.

Next ,we keep $\epsilon>0$ fixed and take $N\rightarrow\infty$.  Then we use the fact that $\epsilon$ is arbitrary to obtain:
 \beq
 \lim_{N \rightarrow \infty}
\frac{1}{d_N}\sum_{j=1}^{d_N}\rho\left(\frac{\tilde\nu_{N,j}}{h^2\epsilon(h)}\right)  =
\lim_{N \rightarrow \infty} \frac{1}{d_N} {\rm Tr} \rho (A_N)
 \eeq

Combining this with Theorem \ref{szego-theorem}, we have now proved that for $\epsilon (h) = h^{q}$, $q>33/2$  and
$h=1/(N+1)$: 
\bea
\lim_{N \rightarrow \infty}
\frac{1}{d_N}\sum_{j=1}^{d_N}\rho\left(\frac{\tilde\nu_{N,j}}{h^2\epsilon(h)}\right)  &=&
 \int_{{\mathcal A}} \rho \left( -\frac{B}{2}
\ell_3 (\alpha)
  \right) ~ d\mu(\alpha) \nonumber \\
&=&\int_{\Sigma
(-1/2) } \rho \left(  -\frac{B}{2}\ell_3 ({\bf x, p})
  \right) ~ d\mu_L ({\bf x, p}).  \phantom{xx}
\eea
Since the eigenvalue shifts  $E_{N,j} (1/(N+1), B) - E_N (1/(N+1))=\frac{\tilde\nu_{N,j}}{h^2}$
(see \eqref{eq:scale2}), this concludes the proofs of Theorem \ref{thm:main1} and Theorem \ref{thm:main2} .


\section{Proof of   Theorem \ref{thm:main3}.}\label{prueba-teorema3}

In order to prove Theorem \ref{thm:main3}, we use an explicit expression for the measure  $d\mu$ obtained in \cite{villegas} in terms of coordinates related to the Kepler problem that we specify below.  Then, by doing the integral indicated in equation (\ref{eq:szego2}), we obtain
equation (\ref{eq:szego3}).


We know that the angular momentum vector $\ell$ and the Runge-Lenz vector $a = {\bf p} \times \ell - \frac{{\bf x}}{|{\bf x}|}$
are integrals of motion for the Kepler problem and satisfy the following relations on the energy surface $\Sigma (-1/2)$:
\bea
\ell^2+a^2=1,   \\
\ell \cdot {a}=0  .
\eea
So we write:
\bea
\ell&=&\cos(\psi)\left(\sin(\theta)\cos(\phi),\sin(\theta)\sin(\phi),\cos(\theta)\right) ,
\\
a&=&\sin(\psi)\left(\cos(\gamma)\hat{u} + \sin(\gamma)\hat{v}\right),
\eea
where $\psi\in(0,\pi/2)$, and $(\theta,\phi)$ are usual spherical coordinates for
the two-sphere $\Sp^2$,  with $\phi\in[0,2\pi]$ and $\theta\in[0,\pi]$.  The unit vectors
$\hat{u}=\left(\sin(\phi),-\cos(\phi),0\right)$ and $\hat{v}=\left(\cos(\theta)\cos(\phi),\cos(\theta)\sin(\phi),-\sin(\theta)\right)$ generate the orthogonal plane to $\ell$.  The parameter
$\gamma\in[0,2\pi]$ takes us around that plane. Note that to cover all possible cases for $|\ell|$ and $|a|$
we should take $\psi\in[0,\pi/2]$, but we only miss
a set of measure zero with the choice $\psi\in(0,\pi/2)$. In this way, the vectors $\ell$ and $a$ are non-zero.
The set $(\psi,\theta,\phi,\gamma)$ gives a
parametrization of the set of
Kepler orbits on the energy surface $\Sigma (-1/2)$. Then, via the the Moser map $\mathcal M$,
we obtain a parametrization for the space of oriented geodesics of the 3-sphere $\Sp^3$  as well.

A parametrization for $\Sigma (-1/2)$, and, consequently, for ${\mathcal A}$ through the Moser map ${\mathcal M}$ and the map $\sigma^{-1}$ (see \eqref{eq:sigma1}), can be obtained by
introducing the parameter $\beta\in[0,2\pi]$ which takes us along the Kepler orbit determined by $\ell$ and $a$.  Namely,  consider $\beta$ such that
\bea
\frac{{\bf x}}{|{\bf x}|} = \cos(\beta)\frac{a}{|a|} + \sin(\beta)\frac{\ell\times{a}}{|\ell||a|}.
\label{ecuac-pos-beta}
\eea
Note that
   $a/|a|$ and $\ell\times{a}/(|\ell||a|)$ give an orthonormal basis for the plane of the Kepler  orbit in configuration space.

   Moreover, it can be shown that, in momentum space $\R^3$, the momentum vector ${\bf p}$ is in a circle with center  $\frac{\ell\times{a}}{|\ell|^2}$  and radius
   $1/|\ell|$. Namely, the coordinate ${\bf p}$
   satisfies the following equation
 \bea
 {\bf p} = \frac{\ell \times{a}}{|\ell|^2} + \frac{\ell\times{{\bf x}}}{|\ell|^2|{\bf x}|}. \label{ecua-mom}
 \eea
Thus from equations (\ref{ecuac-pos-beta}) and (\ref{ecua-mom}) we have
   \bea
  {\bf  p} = -\sin(\beta)\frac{a}{|a||\ell|} + \left(|a| + \cos(\beta)\right).
   \frac{\ell\times{a}}{|a||\ell|^2} \label{ecuac-mom-beta}
   \eea
   From equations (\ref{ecuac-pos-beta}),  (\ref{ecuac-mom-beta}),
  the relation $|{\bf x}|=2/(|{\bf p}|^2+1)$,   the Moser map (see equations (\ref{stereo}) and (\ref{stereoxi})) and the map $\sigma^{-1}$,
 we obtain a parametrization for  ${\mathcal A}$ based on  the five real parameters
 $W=(\psi,\theta,\phi,\gamma,\beta)$.  Namely, we can obtain $(\omega(W),\xi(W)\in{T}^*\Sp^3_0$ (see equations  (4.35), (4.36) and (4.39) in reference \cite{villegas}) such that $\alpha=\omega(W)-\imath{\xi}(W)$, with $\alpha\in\tilde{\mathcal A}$.

  In addition, we can give a parametrization for the elements of the matrix representation by $4\times{4}$ matrices of the group $SO(4)$.    For $g\in{SO(4)}$,
we construct the matrix
\beq\label{para-so4}
{\mathcal R(g)}=\left(\omega(W),\xi(W),
\cos(\delta)\eta(W)+\sin(\delta)\chi(W),-\sin(\delta)\eta(W)+\cos(\delta)\chi(W)\right),
\eeq
from column vectors given by the vectors $\eta(W)$ and $\chi(W)$ that are two orthonormal vectors generating the   plane orthogonal to the vectors $\omega(W)$ and $\xi(W)$.  The parameter $\delta\in[0,2\pi]$ takes us  around the plane generated by  $\{ \eta(W),\chi(W) \}$.

The normalized Haar measure $d\mu_H$ for the group SO(4) can be computed from the parametrization indicated in equation \ref{para-so4}. After a long computation (see reference \cite{villegas} ), we obtain:
\beq\label{haar}
d\mu_H = \frac{1}{(2\pi)^4}\frac{\cos^2(\psi)\sin(\psi)\sin(\theta)}{1+\sin(\psi)\cos(\beta)}d\psi{d}
\theta{d}\phi{d}\gamma{d}\beta{d}\delta.
\eeq
Integrating the measure $d\mu_H $ with respect to $\delta$, we obtain a normalized SO(4) invariant measure $d \mu$ parametrized by $W$:
\bea
d\mu = \frac{1}{(2\pi)^3}\frac{\cos^2(\psi)\sin(\psi)\sin(\theta)}{1+\sin(\psi)\cos(\beta)}d\psi{d}\theta{d}
\phi{d}\gamma{d}\beta.
\eea
Moreover, integrating $d\mu$ with respect to $\beta$,  we obtain the
normalized SO(4)-invariant measure $d\mu_\Gamma$ on the space of oriented geodesics:
\beq\label{eq:mu-gamma1}
d\mu_\Gamma = \frac{1}{(2\pi)^2}\cos(\psi)\sin(\psi)\sin(\theta) d\psi{d}\theta{d}\phi{d}\gamma.
\eeq

Now we want to do the integral appearing on the right hand side of equation
(\ref{eq-main-thm2}). Since the function
$\rho \left( -\frac{B}{2} \ell_3 (\alpha)\right)=\rho\left(-\frac{B}{2}\cos(\psi)\cos(\theta)\right)$  depends only on the variables $\psi$ and $\theta$, we integrate out the variables $\phi$ and $\gamma$ in $d \mu_\Gamma$
in \eqref{eq:mu-gamma1} to obtain the measure $d {\tilde \mu}_\Gamma=\cos(\psi)\sin(\psi)\sin(\theta) d\psi{d}\theta$. Making the change of variables $u=\cos(\psi)\cos(\theta)$, $v=\cos(\psi)\sin(\theta)$, with $(u,v)$ in the
half-disk $u^2+v^2\leq1$, $v\geq{0}$,  we have
\bea
\int_{{\mathcal A}} \rho \left( -\frac{B}{2} \ell_3 (\alpha)
  \right)~ d\mu(\alpha)  &=& \int_{\psi=0}^{\pi/2}\int_{\theta=0}^{\theta=\pi}\rho\left(-\frac{B}{2}\cos(\psi)\cos(\theta)\right)d{\tilde \mu}_\Gamma (\psi, \theta), \nonumber \\
 &=& \int_{u=-1}^{u=1}\rho\left( -\frac{B}{2}u\right)\int_{v=0}^{v=\sqrt{1-u^2}}
 \frac{v}{\sqrt{u^2+v^2}} \;\; dv du,  \nonumber \\
  &=& \int_{-1}^{1}\rho\left( -\frac{B}{2}u\right)\left[1-|u|\;\right] \; du .
\eea
This completes the proof of Theorem \ref{thm:main3}.



\section{Alternate description of the limit measure in Theorem \ref{thm:main3}.}\label{sec:eigen-approx}

The proof of Theorem \ref{thm:main3} presented in section \ref{prueba-teorema3} is based on the geometric description of the Kepler orbits on $\Sigma (-1/2)$ afforded by the Moser map. In this section, we present an alternate proof
based on a more detailed analysis of the eigenvalue clusters $\mathcal{C}_N$. Although this proof is more direct and shorter, it does not reveal the full geometric content of Theorems \ref{thm:main1} and \ref{thm:main2}.

\subsection{Eigenvalue approximation.}\label{eigen-approx}

From equation (\ref{est-normalized}), we know that, for $N$ sufficiently large,  the $d_N=(N+1)^2$ eigenvalues of $S_V(\lambda)$ inside the cluster ${\mathcal C}_N$ around $E_N$ are actually in an interval
of size $O(h^2\epsilon(h))$. In this section, we prove that if we take $\sigma>1$ (i.e.\ $q>19$),
then those eigenvalues are in sub-clusters around the eigenvalues of the operator
$ \Pi_N\left(- \frac{1}{2}\Delta - \frac{1}{|x|}  - \frac{\lambda(h,B)}{2} L_3\right)\Pi_N$ inside the cluster ${\mathcal C}_N$. That is, the eigenvalues cluster around the numbers $E_N - \frac{\lambda(h,B)}{2}m=E_N-\frac{B}{2}h^3\epsilon(h)m$, with $m\in\Z$
and  $|m|\leq{N}$.

\begin{proposition}\label{subclusters}
Assume $\sigma>1$.  Then, for $N$ sufficiently large,  the spectrum of $S_V(\lambda)$ inside the cluster ${\mathcal C}_N$ consists of a union of sub-clusters ${\mathcal C}_{N,m}$ of eigenvalues around $E_N-\frac{B}{2}h^3\epsilon(h)m$, with $m\in\Z$ and  $|m|\leq{N}$. Moreover, each sub-cluster ${\mathcal C}_{N,m}$ contains
  $N+1-|m|$ eigenvalues  of $S_V(\lambda)$ and has size   $h^{\sigma+2}\epsilon(h)$ .
\end{proposition}

\begin{proof}
1. The eigenvalues of the operator $\Pi_N\left(-\frac{B}{2}hL_3\right)\Pi_N$ are
$-\frac{B}{2}\frac{m}{N+1}$, $m\in\Z$ and $|m|\leq{N}$, with multiplicity
$N+1-|m|$.
They are uniformly distributed with a distance
$\frac{B}{2}\frac{1}{N+1}=O(N^{-1})$ between consecutive eigenvalues. Since the error term in equation (\ref{est-normalized}) is a bounded operator whose norm is
$O(N^{-\sigma})$, $\sigma>1$, one can show that  if
the distance between $ w\in\C$ and the spectrum of $\Pi_N\left(-\frac{B}{2}hL_3\right)\Pi_N$ is greater than $O(N^{-\sigma})$, then $w$ is in the resolvent set of $\;\frac{P_N(S_V(\lambda)-E_N)P_N}{h^2\epsilon(h)}\;$.  Namely,
the spectrum of $\;\frac{P_N(S_V(\lambda)-E_N)P_N}{h^2\epsilon(h)}\;$ is contained in the union of the disjoint closed intervals $[-\frac{B}{2}\frac{m}{N+1}-O(N^{-\sigma}),-\frac{B}{2}\frac{m}{N+1}+O(N^{-\sigma})]$, $|m|\leq{N}$.

\noindent
2. Let ${\Gamma}_{N,m}$ be the circle with center $-\frac{B}{2}\frac{m}{N+1}$
and radius $\frac{B}{8}\frac{1}{N+1}$, $|m|\leq{N}$. Let us denote by $\Pi_{N,m}$ and $P_{N,m}$ the Riez projectors asociated to the spectrum of
$\Pi_N\left(-\frac{B}{2}hL_3\right)\Pi_N$ and $\frac{P_N(S_V(\lambda)-E_N)P_N}{h^2\epsilon(h)}$ inside the circle  ${\Gamma}_{N,m}$, respectively.
Then we have
\bea
\|P_{N,m}-\Pi_{N,m}\| =
\left\|\frac{1}{2\pi\imath}\int_{{\Gamma}_{N,m}}\; \left[\left(\frac{P_N(S_V(\lambda)-E_N)P_N}{h^2\epsilon(h)}-\z\right)^{-1}
\right.\right.\nonumber \\
\left.\left. -
\left(\Pi_N\left(-\frac{B}{2}hL_3\right)\Pi_N-\z\right)^{-1} \right] d \z \right\|.
\label{difprojm}
\eea
We note that both $\left\|  \left(\frac{P_N(S_V(\lambda)-E_N)P_N}{h^2\epsilon(h)}-\z\right)^{-1} \right\|$ and $\left\| \left(\Pi_N\left(-\frac{B}{2}hL_3\right)\Pi_N-\z\right)^{-1}  \right\|$ are $O(N)$. From equations
(\ref{est-normalized}) and (\ref{difprojm}), we obtain
$\|P_{N,m}-\Pi_{N,m}\| = O(N^{1-\sigma})$, which implies that,
for $N$ sufficiently large,
$\|P_{N,m}-\Pi_{N,m}\|<1$. Hence  the dimension of the range of the projectors  $P_{N,m}$ and
$\Pi_{N,m}$ is the same, namely equal to $N+1-|m|$.
\end{proof}

\noindent
Proposition \ref{subclusters} implies the following eigenvalue approximation.

\begin{proposition}\label{prop:eigen-approx}
 Assume $\sigma>1$.  The  eigenvalues  of $S_V(\lambda)$ inside  the cluster ${\mathcal C}_N$ around $E_N$ can be written in the following way:  For $N$ sufficiently large and
$m=-N,\ldots,N$,
we have
\bea
E_{N,m,k} = E_N  - \frac{B}{2}\frac{m}{N+1}h^2\epsilon(h) + G(N,m,k)
\eea
where, given $m$,  the index $k=1,\ldots,N+1-|m|$ and
the error term $G(N,m,k)=O(N^{-\sigma})h^2\epsilon(h)$.
\end{proposition}

\begin{remark} A similar expansion was obtained by Karasev and Novikova \cite{karasev-novikova} using different methods.
\end{remark}

\subsection{Alternate proof of Theorem \ref{thm:main3}.}

Using Proposition  \ref{prop:eigen-approx}, we can prove Theorem \ref{thm:main3} when
$\rho$ is a polynomial on a fixed compact interval. Then Theorem \ref{thm:main3} can  be shown for $\rho$ continuous with a uniform approximation argument.

The averages appearing in the left hand side of equation
(\ref{eq-main-thm3}) can be written, for $N$ sufficiently large,   as:
\bea
&&\frac{1}{d_N}
\sum_{j=1}^{d_N} \rho \left( \frac{ E_{N,j} (1/(N+1), B) - E_N (1/(N+1)) }{ \epsilon
(1/(N+1)) } \right) \phantom{xxxxxxxxxxxxxxxxx} \nonumber \\
&&\phantom{xxxx}= \frac{1}{d_N}\sum_{m =-N}^{m=N}\sum_{k=1}^{N+1-|m|}\rho\left( - \frac{B}{2}\frac{m}{N+1} + \frac{G(N,m,k)}{h^2\epsilon(h)}\right) \nonumber \\
&&\phantom{xxxx}=
\frac{1}{d_N}\sum_{m =-N}^{m=N}\left(N+1-|m|\right)\rho\left(- \frac{B}{2}\frac{m}{N+1}\right) + O(N^{-\sigma}) \nonumber \\
&&\phantom{xxxx}= \sum_{m =-(N+1)}^{m=N}\left(1-\frac{|m|}{N+1}\right)
\rho\left(- \frac{B}{2}\frac{m}{N+1}\right)\frac{1}{N+1} + O(N^{-\sigma}).
 \label{riemann-sum}
\eea
To obtain the second equality, we have used the mean value theorem,
that $\rho^\prime$ is bounded on a fixed compact interval, and the fact that
$\sum_{m =-N}^{m=N}\left(N+1-|m|\right)=d_N$.   The first term in equation (\ref{riemann-sum}) can be thought of as a Riemann sum associated to $\int_{[-1,1]}\rho(-\frac{B}{2}u)\left(1-|u|\right)du$ with the partition $\left\{-1,\frac{-N}{N+1}, \ldots, \frac{N}{N+1},1\right\}$ of the interval $[-1,1]$.  Taking the limit $N\rightarrow\infty$ we conclude the proof of Theorem \ref{thm:main3}.


\section{Appendix 1: The Kepler problem and the Moser map.}\label{sec:moser-map}

We review the Moser map as a regularization of the Kepler problem.
The Kepler problem in $\mathbb{R}^3-\{0\}$ is defined as the
Hamiltonian $G:T^*(\mathbb{R}^3-\{0\})\rightarrow\R$
given by
\beq
G({\vecx, \vecp})= \frac{|{\bf p}|^2}{2}-\frac{1}{|{\bf x}|}.
\eeq

The symmetries of the Kepler problem are given by conservation of
the angular momentum vector $\vecell = \vecx \times {\vecp}$ and the Runge-Lenz
vector $\veca =\vecp\times{\vecell}-\frac{\vecx}{|\vecx|}$. For negative energy, they
imply that the orbits in configuration space are either ellipses
($|\vecell|\neq{0}$) or contained in segments of straight lines with the
origin $\vecx=0$ as an end ($\vecell=0$, collision orbits).  The symmetries
$\vecell$ and $\veca$ also imply that the orbits in momentum space are
either circles ($|\vecell|\neq{0}$) or half lines passing through the origin ($\vecell=0$).

We now restrict ourselves to the energy surface $\Sigma(-1/2)=
\{({\bf x,p})\in{T}^*(\R^3-\{0\})\; |
\; G({\bf x,p}) = - 1/2 \}.$
We first note that the vectors $\vecell$ and $\veca$ satisfy the relation
$|\vecell|^2+|\veca|^2=1$ on $\Sigma(-1/2)$ which, in particular, implies that $|\ell_3|\leq{1}$  on $\Sigma(-1/2)$ with $\vecell=(\ell_1,\ell_2,\ell_3)$.

Following J.\ Moser \cite{Moser}, we consider the stereographic projection $S:{\mathbb R}^3\rightarrow{\Sp}^3_o$, where  $\Sp^3_o$ denotes the 3-sphere
with the north pole removed.   The map $S$ is given by the assignment  ${\bf p} =(p_1,p_2,p_3) \rightarrow \omega=(\omega_1,\omega_2,\omega_3,\omega_4)$ with
\bea\label{stereo}
\omega_i &=& \left(\frac{2}{|{\bf p}|^2+1}\right)p_i, \;\;\;  i=1,2,3,   \;\;\; |{\bf p}|^2=p_1^2+p_2^2+p_3^2   \nonumber \\
\omega_4 &=& \frac{|{\bf p}|^2-1}{|{\bf p}|^2+1}.
\eea

 The stereographic projection $S$ maps the circles in momentum space mentioned above to great
circles on $\Sp^3$ that do not contain the north pole. The image of a
half line passing through the origin in momentum space under the stereographic projection is
contained in a great circle passing through the north pole.

Moser \cite{Moser} extended the stereographic projection  $S$  to a
map ${\mathcal M}:T^*({\mathbb R}^3)\rightarrow{T}^*({\Sp}^3_o)$. The Moser map ${\mathcal M}$
maps
the point $({\bf x,p})$ to $(\omega,\vecxi)$  under the requirement
that 
${\mathcal{M}}^*(\vecxi\cdot {d}\omega)= {\bf y} \cdot  {d}{\bf p}$ with ${\bf y}=-{\bf x}$. Thus the Moser map is a
canonical transformation ${\mathcal M}^*({d}\vecxi\wedge{d}\omega)={d} {\bf p} \wedge {d}{\bf x}$ given explicitly
by (\ref{stereo}) and the following equations:
\begin{eqnarray}
\xi_i&=& \frac{|{\bf p}|^2+1}{2}y_i - ({\bf y} \cdot {\bf p})p_i, \;\;\;
i=1,2,3 \nonumber \\
\xi_{4}&=& {\bf y} \cdot {\bf p}. \label{stereoxi}
\end{eqnarray}

The inverse map ${\mathcal M}^{-1}$ is determined by the equations
\begin{eqnarray}
p_i&=&\frac{\omega_i}{1-\omega_{4}},  \;\;\;
i=1,2,3 \nonumber \\
y_i&=& (1-\omega_{4})\xi_i+ \xi_{4}\omega_i.   \;\;\; i=1,2,3
\label{invmoser}
\end{eqnarray}


The Moser map ${\mathcal M}$ transforms the Hamiltonian flow of
the Kepler problem on the energy surface $\Sigma(-1/2)$ into the geodesic flow on
the 3-sphere under the time re-parametrization $s\mapsto{t}$ given
by the equation
\begin{eqnarray}
\frac{d}{ds}=|{\bf x}|\frac{d}{dt} \label{trep}.
\end{eqnarray}
Considering the geodesic flow on $T^*(\Sp^3)$ (i.e.\ including the
north pole) corresponds to extending the
collision orbits by making the convention that the particle is
reflected back to its trajectory after a collision.  Thus all of
the orbits on the energy surface $\Sigma(-1/2)$ are periodic
orbits with period $2\pi$. See reference \cite{Moser} for details.




\section{Appendix 2: Coherent states for the hydrogen atom.}\label{sec:app2-coherentstates1}

We review the construction and properties of the coherent states
that form an over-complete set in the eigenspace $\mathcal{E}_N$
of the hydrogen atom Hamiltonian $S_V = -\frac{1}{2}\Delta - \frac{1}{|x|}$
corresponding to the eigenvalue
$E_N = - 1/ (2  (N+1)^2)$,   $N\in\N^*$.

We begin by defining    the null quadric ${\mathbb Q}^n$, with $n\geq{1}$ a natural number, as the set
\bea
{\mathbb Q}^n &=& \left\{ \alpha=(\alpha_1\ldots,\alpha_{n+1})\in \C^{n+1}\;\; | \;\;  \alpha_1^2+\ldots+\alpha_{n+1}^2=0\right\} ,\nonumber \\
&=& \left\{ \alpha\in\C^{n+1} \;\; | \;\;  |\Re \alpha|=|\Im \alpha| ,   \;\;  \Re\alpha
 \cdot \Im \alpha =0 \right\}.
\eea
The null quadric ${\mathbb Q}^n$ with the origin removed can be identified with the cotangent bundle of the n-sphere $T^*\Sp^n$ with the zero section removed through the following map :
\bea\label{eq:sigma1}
\sigma&:& {\mathbb Q}^n-\{0\} \rightarrow T^*\Sp^n-\{0\} \nonumber \\
\sigma(\alpha) &=& \left(-\frac{\Re\alpha}{|\Re\alpha|},\Im\alpha\right).
\eea

Let $\mathcal{A}$ be the $2n-1$
real-dimensional subset of ${\mathbb Q}^n$ defined by
\beq\label{eq:defnA1}
\mathcal{A} = \{ \alpha\in{\C}^{n+1} ~|~  ~| \Re
\alpha | = | \Im \alpha | = 1, ~ \Re \alpha \cdot  \Im \alpha = 0 \} \subset {\mathbb Q}^n .
\eeq
This provides a parametrization of the
unit cotangent  bundle $T_1^* \Sp^n$ of the n-sphere  $\Sp^n$ under the map $\alpha\rightarrow (\Re \alpha , - \Im \alpha )$. There is a
$SO(n+1)$-rotationally invariant probability measure on
$\mathcal{A}$ that we denote by $\mu$.

Coherent states on $\Sp^n$
have the form
\beq\label{eq:cost-defn1}
\Phi_{\alpha,N}(\omega) = a(N) ( \alpha \cdot \omega)^N, ~~\omega \in
\Sp^n, ~~ \alpha \in \mathcal{A}, ~~N\in\N^*. \eeq
The coefficient $a (N) \sim N^{(n-1)/2}$ is fixed by the
requirement that the $L^2 ( \Sp^n)$-norm of $\Phi_{\alpha, N}$ is
equal to one, see \cite[(2.11)]{thomas-villegas1}. The states $\Phi_{\alpha,N}$
are spherical harmonics i.e.  they are eigenstates of the
normalized spherical Laplacian $\Delta_{\Sp^n}$ (the usual positive spherical Laplacian on the sphere plus the constant $(n-1)^2/4$) with eigenvalue
$(N+\frac{n-1}{2})^2$ . The entire family $\{ \Phi_{\alpha, N}(\omega) ~|~
\alpha \in \mathcal{A} \}$ is over complete and spans the
eigenspace ${\mathcal L}_N$ of $ \Delta_{\Sp^3}$ with eigenvalue $(N+\frac{n-1}{2})^2$ in the following sense:
  Let $P_N^S$ be the projector from $L^2(S^n)$ onto ${\mathcal L}_N$.   Then
  for all $\Psi\in{L}^2(S^n)$
\beq
P_N^S\Psi = d_n(N) \int_{\mathcal{A}} ~ \langle\Phi_{\alpha, N},\Psi\rangle_{L^2(S^n)}  \;\; \Phi_{\alpha, N}
 ~d \mu(\alpha) . \label{sph-res}
\eeq
where $d_n(N)$ denotes the dimension of ${\mathcal L}_N$. The following notation for the projector $P_N^S$ in equation
\eqref{sph-res} is also used:
\beq
P_N^S =  d_n(N) \int_{\mathcal{A}} ~ |\Phi_{\alpha, N}\rangle\langle\Phi_{\alpha, N}|  ~d \mu(\alpha) .
\eeq

We note that the state $\Phi_{\alpha, N}$  concentrates on the great circle $\{
\omega \in \Sp^n ~|~ | \alpha \cdot \omega | = 1 \}$ generated by $\alpha$ as $N \rightarrow \infty$.

The normalization
factor $a(N)$ can be estimated by the stationary phase method.  Here we show an estimate of the error $O(N^{-1})$ which improves the estimate $O(N^{-1/2})$ obtained in \cite[(2.11)]{thomas-villegas1}.

\begin{proposition}\label{estnormacs}
For $N$ large, we have \beq a^2(N) =
\left(\frac{N}{\pi}\right)^{\frac{n-1}{2}}\left[\frac{1}{2\pi}+O(N^{-1})\right].
\eeq
\end{proposition}

\begin{proof} The proof follows a  similar procedure as the one
used to show \eqref{prodint}. Namely,  given $\alpha$ in
${\mathcal A}$, let us consider the following coordinates for the
n-sphere. Consider orthonormal vectors
$\vece_3,\ldots,\vece_{n+1}$ such that
$\{\Re{\alpha},\Im{\alpha},\vece_3,\ldots\vece_{n+1}\}$ is an
orthonormal basis of $\mathbb{R}^{n+1}$.  Then for almost every
element $\omega\in\Sp^n$ we can write
$\omega=\sqrt{1-|\vecz|^2}\cos(\theta)\Re{\alpha}+\sqrt{1-|\vecz|^2}\sin(\theta)\Im{\alpha}+z_3\vece_3+
\ldots+ z_{n+1}\vece_{n+1}$ with $\theta\in[0,2\pi]$ and
$\vecz=(z_3,\ldots,z_{n+1})$ such that $|\vecz|^2<1$.  The volume form in
these coordinates is $d\Omega(\omega)={d}z_3 \cdots {d}z_{n+1} d \theta$.
Thus \bea
\int\left|\alpha\cdot{\omega}\right|^N{d}\Omega(\omega) =
\int_{0}^{2\pi}\int_\vecz\exp\left(\imath{N}\phi_{\alpha}(\vecz)\right)dz_3\cdots{d}z_{n+1}d\theta
\phantom{xxxxxxxxxxx} \nonumber \\ =
2\pi\left[\det\left(\frac{\phi_\alpha^{\prime\prime}(0)N}{2\pi\imath}\right)\right]^{-1/2}
+ O(N^{-\frac{n+1}{2}}) =
\left(\frac{\pi}{N}\right)^{\frac{n-1}{2}}\left[2\pi +
O(N^{-1})\right]
 \eea
where $\phi_{\alpha}(\vecz)=-\imath\ln(1-|\vecz|^2)$ and
$\phi_\alpha^{\prime\prime}(0)$ denotes the $(n-1)\times(n-1)$
Hessian matrix of $\phi_{\alpha}$ evaluated at the critical point
$\vecz=0$ such that
$\det\left(\frac{\phi_\alpha^{\prime\prime}(0)N}{2\pi\imath}\right)=
\left(\frac{N}{\pi}\right)^{n-1}$.
\end{proof}

Let us now restrict ourselves to the $n=3$ case.
The coherent states functions in both momentum  and configuration space, constructed from the coherent states for the 3-sphere defined in equation (\ref{eq:cost-defn1}), were considered by  Thomas and Villegas-Blas \cite{thomas-villegas1}.  The construction uses a
transformation due to V.\ Fock \cite{Fo} and described in reference \cite{BI:1966}.

The coherent states functions in momentum space $\hat{\Psi}_{\alpha, N}$ are defined as a suitable dilation of $\Phi_{\alpha, N}(\omega(\vecp))$, where $\omega(\vecp)$ is given by the stereographic projection map $S$ defined in equations (\ref{stereo}), times a factor that involves not only the square root of the corresponding Jacobian but also a  factor $2/(|\vecp|^2+1)$, which is precisely $|\vecx|$ on the energy surface $\Sigma(-1/2)$ of the Kepler problem. In addition, the factor $|\vecx|$ is exactly  the factor used in the time reparametrization for the regularization of the Kepler problem.  We still have the isometry property  $\|\hat{\Psi}_{\alpha, N} (\vecp)\|=1$ due to the Virial Theorem.  See reference \cite{BI:1966} for details.

The coherent states $\hat{\Psi}_{\alpha, N}$ have the following explicit
form.
For any
$\alpha \in \mathcal{A}$, we define \beq\label{eq:cohst-p1}
\hat{\Psi}_{\alpha, N} (\vecp) = a(N) (N+1)^{3/2} \left(
\frac{2}{(N+1)^2 |\vecp|^2 +1} \right)^2 (\alpha \cdot \omega
((N+1) \vecp) )^{N},
 ~p  \in \R^3 .
\eeq
The  coherent states  $\hat{\Psi}_{\alpha, N}$ are in $L^2 ( \R^3)$.

To obtain the formula for the coherent states in configuration space, we
 use the Fourier transform
given by
\beq\label{defft}
{\mathcal F}(g)(\vecp)=\frac{1}{(2\pi)^{3/2}}\int ~e^{-\imath{\vecp}\cdot{\vecx}} ~g(\vecx) ~d\vecx,
\eeq
that is a unitary operator on $L^2 (\R^3)$.
The inverse Fourier transforms ${\Psi}_{\alpha, N}$ of the coherent states
$\hat{\Psi}_{\alpha, N}$ are eigenfunctions of $S_V$ with
eigenvalue $E_{N}=-\frac{1}{2(N+1)^2}$.
For fixed $N$ and $\alpha \in \mathcal{A}$,
they form a normalized (but not orthogonal) overdetermined basis
of the $(N+1)^2$-dimensional eigenspace $\mathcal{E}_{N}
\subset L^2 (\R^3)$ of the hydrogen atom Hamiltonian with
eigenvalue $E_{N}$. These coherent states functions in configuration space
have the form
\bea\label{eq:cohst-p2} {\Psi}_{\alpha, N} (\vecx)&=&
\frac{a(N) (N+1)^{3/2}}{ (2 \pi)^{3/2}} \int_{\R^3} ~d \vecp ~e^{i \vecx \cdot \vecp} ~\left\{
~\left( \frac{2}{(N+1)^2 |\vecp|^2 +1} \right)^2 \right. \nonumber \\
&& \times
 \left. (\alpha \cdot \omega ((N+1) \vecp) )^{N} \right\}. \nonumber \\
\eea

Note that the coherent states ${\Psi}_{\alpha, N}$ can be written as
\beq\label{cohstaeshydatom} {\Psi}_{\alpha, N} =
{\mathcal
F}^{-1}D_{r_{N}}J^{1/2}K\Phi_{\alpha, N} \eeq where
$r_{N}=N+1$,
 $K:L^2(S^3)\rightarrow{L^2(\mathbb{R}^3)}$ is the unitary
operator \beq\label{defK} K(f)(\vecp) =
\left(\frac{2}{|\vecp|^2+1}\right)^{3/2}f(\omega(\vecp)) \eeq and J is
the self-adjoint multiplicative operator acting in momentum space
$J:L^2(\mathbb{R}^3)\rightarrow{L^2(\mathbb{R}^3)}$ given by
\beq\label{defJ}
J(\hat{\Psi})(\vecp)=\frac{2}{|\vecp|^2+1}\hat{\Psi}.(\vecp) \eeq

We remark that $d\Omega({{\omega}})=\left(\frac{2}{|\vecp|^2+1}\right)^3d\vecp$ under the stereographic  projection (\ref{stereo}).  Thus the factor $\left(\frac{2}{|\vecp|^2+1}\right)^{3/2}$ defining the operator $K$ is the square root of the Jacobian $\left(\frac{2}{|\vecp|^2+1}\right)^3$.

Let $\mu$ be the normalized SO(4)-invariant  measure on $\mathcal{A}$. The orthogonal
projector $\Pi_N$ onto the eigenspace $\mathcal{E}_N$ can
be written as \beq\label{eq:proj-coherentstates1} \Pi_N  =
(N+1)^2 \int_{\mathcal{A}} ~| \Psi_{\alpha, N} \rangle
\langle \Psi_{\alpha, N}| ~d \mu(\alpha) . \eeq
Equation (\ref{eq:proj-coherentstates1}) expresses the fact that the system of coherent states  $\{\Psi_{\alpha,N}\}_{\alpha\in{\mathcal A}}$ gives a resolution of the identity for the eigenspace $\mathcal{E}_N$.

In reference \cite{thomas-villegas1}, it is shown that  the coherent state  ${\Psi}_{\alpha, N}$
concentrates on the Kepler orbit determined by $\alpha$ as
$N \rightarrow \infty$ in the case when such an orbit is an ellipse.



\end{document}